\documentclass[smallextended]{svjour3}  
\usepackage{amssymb}
\usepackage{amsmath}
\usepackage{graphicx}
\usepackage{epstopdf}
\usepackage{color}
\begin{document}

\title{Dynamic boundary conditions for membranes whose surface energy depends on the mean and Gaussian curvatures}

\subtitle{}

\titlerunning{Energy depending on the mean and Gaussian curvatures}        

\author{Sergey Gavrilyuk\  \ and \,  Henri Gouin}

\institute{ \at Aix Marseille Univ, CNRS, IUSTI, \  UMR 7343,  Marseille, France\\  \email \\  sergey.gavrilyuk@univ-amu.fr,  \\     {henri.gouin@univ-amu.fr; henri.gouin@yahoo.fr}     
      }
\date{Accepted: 21-02-2019 - Mathematics and Mechanics of Complex Systems (ISSN : 2326-7186, ESSN : 2325-3444)}

\maketitle

\begin{abstract}
	 
	Membranes are an   important subject of study in physical chemistry and  biology. They can be  considered as material surfaces with a  surface energy  depending  on the curvature tensor. Usually,   mathematical models developed in the literature consider the  dependence of surface energy  only on  mean curvature with an added linear term for  Gauss curvature.  
	Therefore, for closed surfaces the Gauss curvature term can be eliminated because of  the Gauss-Bonnet theorem. In \cite{virga}, the dependence on the mean and Gaussian curvatures was considered in statics   and under a restrictive assumption of the membrane inextensibility. The authors derived the shape equation as well as two scalar boundary conditions on the contact line. 
	
	In this paper -- thanks to the principle of virtual working -- the equations of motion and boundary conditions governing
	the fluid membranes subject to general dynamical bending are derived without  the membrane inextensibility assumption. We   obtain the dynamic `shape equation' (equation for the membrane surface) and the dynamic conditions on the contact line  generalizing  the  classical Young-Dupr\'e condition.    
	
\end{abstract}

\section{Introduction}

The study of equilibrium, for small wetting droplets placed on a curved rigid surface, is an old problem of continuum mechanics. When the droplets' size is  of micron range  the droplet volume energy can be neglected.  The surface energy of the surface $S$ can be expressed in the form :
\begin{equation*} 
E=\iint_S\sigma \; ds,
\label{_very_general_energy}
\end{equation*}
where $\sigma$ denotes the energy per unit surface. Two types of surfaces are present in physical problems:
\begin{itemize}
\item rigid surfaces (only the kinematic boundary condition is imposed)
\item free surfaces (both the kinematic and dynamic boundary conditions are imposed)
\end{itemize}
 We will  see the difference between the energy variation in the case of rigid and free surfaces. 
 
 The simplest case corresponds to a constant surface energy $\sigma$, but in general,  $\sigma$  also  depends on physical parameters (temperature, surfactant concentrations, etc.  \cite{Gouin_2014,Rocard,Steigmann1}) and geometrical parameters (invariants of  curvature tensor). 
 The last case is important in biology  and,  in particular, in the dynamics of {\it vesicles} \cite{Alberts,Lipowsky,seifert}. Vesicles are small liquid droplets with a diameter of a few tens of micrometers, bounded by an impermeable lipid membrane of a few nanometers thick.  The membranes are homogeneous down to molecular dimensions. Consequently, it is possible to model the boundary of vesicle as a two-dimensional smooth surface  whose energy per unit surface $\sigma$ is a function  both of the sum (denoted by $H$) and product (denoted by $K$) of principal curvatures of  the curvature tensor : 
 \begin{equation*}
 \sigma=\sigma (H, K).
 \end{equation*}
In  mathematical description of biological membranes, one often uses the Helfrich energy \cite{Helfrich,Tu} :
\begin{equation} 
\sigma (H,K)=\sigma_0+\frac{\kappa}{2}(H-H_0)^2+\bar\kappa K,
\label{Helfrich_energy}
\end{equation}
where $\sigma_0$,  $H_0$, $\kappa$  and $\bar\kappa$ are dimensional constants. 
Another purely mathematical example is the Wilmore energy  \cite{Willmore} :
\begin{equation*} 
\sigma (H,K)=H^2-4K.
\label{Willmore_energy}
\end{equation*}
This   energy measures  the ``roundness" of   the free surface. For  a given volume, this energy is minimal in case of  spheres. 
One can also propose another surface energy in the form :
\begin{equation*} 
\sigma=\sigma_0+h_0(H^2-H_0^2)^2+k_0 (K-K_0)^2,
\label{our_energy}
\end{equation*}
where $\sigma_0$, $h_0$, $H_0$, $k_0$ and $K_0$ are dimensional constants.
This kind of  energy is invariant under the change of   sign of   principal curvatures, ({\it i.e.} the change of sign yields  $\; H\rightarrow -H,\; K\rightarrow K$). It can thus  describe the `mirror buckling' phenomenon : a portion of the membrane inverts to form a cap with equal but opposite principal curvatures.   It is also a homogeneous function of degree four with respect to  principal curvatures.

The equilibrium for membranes (called ``shape equation"  by Helfrich) is formulated   in numerous papers and references herein \cite{Biscari,Capovilla,Fournier,Helfrich,Napoli,Helfrich1}. The ``edge conditions" (boundary conditions at the contact line) are formulated in few papers and only in statics.  In particular,  in \cite{virga} the shape equation and two boundary conditions are formulated for the general dependence $\sigma (H,K)$ under the   assumption of the membrane inextensibility.   However, the boundary conditions obtained do not contain the classical Young-Dupr\'{e}   condition for the constant surface energy.  In the case when the energy depends only on $H$  the generalization of Young-Dupr\'{e} condition   was  obtained in  \cite{Gouin_2014_vesicle}.

The aim of our paper is to develop the theory of moving membranes  which are  in contact   with a  solid surface.  The surface energy of the membrane will be a function both of $H$ and $K$. We  obtain a set of  boundary conditions on the  moving interfaces (membranes)  as well as on the moving edges. 

The motion of a continuous medium is 
represented by a diffeomorphism  $\boldsymbol{\phi}$ of a
three-dimensional reference configuration $D_0$ into the physical space. In order to analytically
describe the transformation,  variables
$\mathbf{X}=(X^{1},X^{2},X^{3})^T$  single out individual
particles corresponding to material or Lagrangian coordinates, subscript ``${^T} $'' means the transposition. The
transformation representing the motion of a continuous medium occupying the material volume $D_t$ is :
\begin{equation*}
\mathbf{x}= {\boldsymbol\phi}\left(t, \mathbf{X}\right) \text{ \ or \
}x^{i}=\phi ^{i}(t, X^{1},X^{2},X^{3})\, , \  i =\{1, 2, 3\} ,
\label{motion}
\end{equation*}
where $t$ denotes the time and $\mathbf{x}=(x^{1},x^{2},x^{3})^T$ denote the Eulerian coordinates.  At $t$ fixed, the transformation
possesses an inverse and has continuous derivatives up to the second order (the dependence of the surface energy on the curvature tensor  will regularize the solutions, so the cusps and shocks do not appear).

At equilibrium, the unit normal vector to a static surface $\varphi_0({\mathbf x})=0$  is the gradient of 
the so-called {\it signed distance function} defined as follows.
Let   
\begin{equation} 
d(\mathbf{x})=
\left\{
\begin{array}{l}\ \ {\rm min}\lvert\mathbf{x}-\boldsymbol{\xi}\rvert, \; {\rm if }\; \varphi_0>0, \\
\ \ 0,\; {\rm if }\; \varphi_0=0, \\
-{\rm min}\lvert\mathbf{x}-\boldsymbol{\xi}\rvert, \; {\rm if }\; \varphi_0<0,\\
\end{array}
\right. 
\label{evolution_interface}
\end{equation}
where the minimum is taken over points $\boldsymbol{\xi}$ at the surface,  and $\lvert\ \rvert$ denotes the Euclidien norm.
The unit normal  vector is : 
\begin{equation*}
\mathbf{n}=\nabla d(\mathbf{x}).
\label{normal vector} 
\end{equation*} 

In  dynamical problems, the main difficulty in formulating  boundary conditions  comes from  the fact that {\it one cannot assume that for all time} $t$ the unit normal vector to the surface is the gradient of 
the signed distance function.

Indeed,  if the material surface is moving, {\it i.e.} the surface position depends  on time $t$, the surface points of the continuum medium are also moving  and they will depend implicitly on $\mathbf{x}$.  Let $\varphi(t, \mathbf{x})=0$  be
  the position of the  material surface at time $t$. Its evolution is determined by the equation :
\begin{equation} 
\varphi_t+\mathbf{u}^T\nabla \varphi =0, 
\label{evolution_interface}
\end{equation}
where  $\mathbf{u}$ is the   velocity of particles at the surface.  Equation \eqref{evolution_interface} is the classical  kinematic condition for material  moving interfaces. Let us derive the equation for the norm of $\nabla\varphi$. Taking the gradient of Eq. \eqref{evolution_interface} and multiplying  by $\nabla\varphi$, one obtains : 
\begin{equation} 
\left(\lvert\nabla\varphi\rvert\right)_t+ \mathbf{n}^T\nabla\left(\mathbf{u}^T\nabla\varphi\right) =0,
\label{surface_area_evolution}
\end{equation}
where $\displaystyle\mathbf{n}=\frac{\nabla\varphi}{\lvert\nabla\varphi\rvert}$ is the unit normal vector to  surface $\varphi(t, \mathbf{x})=0$. It follows from Eq. \eqref{surface_area_evolution} that, even  if initially $\lvert\nabla\varphi\rvert=1$ ({\it i.e.} unit normal $\mathbf{n}$ is defined at $t=0$ as the gradient of the  signed distance function), this property is not conserved in time. 

The following definitions and notations are used in the paper.    
For  any vectors 
$\mathbf{a}, \mathbf{b}$, we write $\mathbf{a}^T\, \mathbf{b}$  for their {\it scalar product} (the
line vector is multiplied by the column vector),  and $\mathbf{a}\, \mathbf{b}^T$  for their {\it tensor product} (the column vector is multiplied by the line vector). The last product is usually denoted as  $\mathbf{a}\otimes\mathbf{b}$. 
The product of a second order tensor    $\mathbf{A}$   by a vector   $\mathbf{a}$
is denoted by   $\mathbf{A}\,  \mathbf{a} $. Notation   $\mathbf{b}^T  \mathbf{A}\, $ means the covector   $\mathbf{c}^T $ defined
by the rule $\mathbf{c}^T = ( \mathbf{A}^T\,  \mathbf{b})^T$. The identity tensor is denoted by $\mathbf{I}$. 
\\
The
divergence of  $\mathbf{A}$   is  
covector $ {\rm div} \mathbf{A}  $ such that, for any constant vector 
$\mathbf{h}$, one has  $$ \left({\rm div} \mathbf{A}\right)\, \mathbf{h}    =   {\rm
div }\, (\mathbf{A}\,\mathbf{h}),$$ i.e. the divergence of $\mathbf A$ is a row vector, in which each component is the divergence of the corresponding column of $\mathbf A$. It implies 
\begin{equation*}
{\rm div} \left(\mathbf{A} \mathbf{v} \right)   =   \left({\rm
	div }\mathbf{A}\right)\mathbf{v}+{\rm tr}\left(\mathbf{A}\frac{\partial \mathbf{v}}{\partial\mathbf{x}}\right),
\end{equation*}
for any  vector field $\mathbf{v}$. Here  ${\rm tr}$ is the trace operator.
If $f$ is a real scalar 
function of $\mathbf{x}$, $ \displaystyle \frac {\partial f}{\partial
\mathbf{x}}$  is the linear form (line vector) associated with the gradient of
$f$ (column vector) : $\displaystyle{ \frac {\partial f}{\partial
		\mathbf{x}}=(\nabla f)^T}$.

If $\mathbf{n}$ is the unit normal vector to a surface,  $\mathbf{P} =\mathbf{I}-\mathbf{n}\mathbf{n}^T$ is the projector on the surface with the classical   properties :  
\begin{equation*}
\mathbf{P}^2=\mathbf{P}, \quad \mathbf{P}^T=\mathbf{P}, \quad \mathbf{P}\mathbf{n}=\mathbf{0}, \quad \mathbf{n}^T\mathbf{P}=\mathbf{0}. 
\end{equation*} 
For any scalar field $f$, the vector field $\mathbf{v}$ and second order tensor field $\mathbf{A}$, the tangential surface gradient, tangential surface divergence, Beltrami--Laplace operator, and tangent tensors  are defined as :
\begin{equation*} 
{\mathbf v}_{\rm tg}=\mathbf{P}{\mathbf v}, \quad\mathbf{A}_{\rm tg}= \mathbf{P}\mathbf{A},\quad \nabla_{\rm tg}f=\mathbf{P}\nabla f, 
\end{equation*}
\begin{equation*}
{\rm div_{tg}} \mathbf{v}_{\rm tg}={\rm tr}\left(\mathbf{P}\frac{\partial \mathbf{v}_{\rm tg}}{\partial \mathbf{x}}\right), \quad \Delta_{\rm tg}f= {\rm div_{tg}}\left(\nabla_{\rm tg}f\right), \label{BelLap}
\end{equation*}
and for any constant vector $\mathbf  h$,
\begin{equation*}
{\rm div_{tg}} \left(\mathbf{A}_{\rm tg}{\mathbf h}\right)={\rm div_{tg}} \left(\mathbf{A}_{\rm tg}\right){\mathbf h}.
\end{equation*}
The following relations between surface   operators and classical operators applied to tangential  tensors in the sense of previous definitions are valid  : 
\begin{eqnarray}
&&{\rm div_{tg}} \mathbf{v}_{\rm tg}={\rm div} \mathbf{v}_{\rm tg}+{\mathbf n}^T\left(\frac{\partial {\mathbf n}}{\partial{\mathbf x}}\right)^T\mathbf{v}_{\rm tg},\label{1}\\
&&{\rm div_{tg}} \mathbf{v}_{\rm tg}={\mathbf n}^T{\rm rot}\left({\mathbf n}\times{\mathbf v}_{\rm tg}\right),\label{2}\\
&&{\rm div_{tg}} \mathbf{A}_{\rm tg}={\rm div} \mathbf{A}_{\rm tg}+{\mathbf n}^T\left(\frac{\partial {\mathbf n}}{\partial{\mathbf x}}\right)^T\mathbf{A}_{\rm tg},\label{3}\\
&&{\rm div_{tg}} \left(f\mathbf{v}_{\rm tg}\right)= f\,{\rm div}_{\rm tg} \mathbf{v}_{\rm tg}+\left(\nabla_{\rm tg}f\right)^T\mathbf{v}_{\rm tg}\label{4},\\
&&{\rm div_{tg}} \left(f\mathbf{A}_{\rm tg}\right)=f\,{\rm div}_{\rm tg} \mathbf{A}_{\rm tg}+\left(\nabla_{\rm tg}f\right)^T\mathbf{A}_{\rm tg}\label{5},
\end{eqnarray}
where ${\rm rot}$  denotes the curl operator.
The proof is straightforward. Indeed, since
\begin{equation*}
\frac{\partial ({\mathbf n}^T{\mathbf v}_{\rm tg})}{\partial {\mathbf x}}={\mathbf n}^T\left(\frac{\partial {\mathbf v}_{\rm tg}}{\partial{\mathbf x}}\right)+{\mathbf v}_{\rm tg}^T\left(\frac{\partial {\mathbf n}}{\partial{\mathbf x}}\right)=0,
\end{equation*}
one has 
\begin{equation*}
{\rm div_{tg}} \mathbf{v}_{\rm tg}={\rm tr}\left(\mathbf{P}\,\frac{\partial \mathbf{v}_{\rm tg}}{\partial \mathbf{x}}\right)={\rm div} {\mathbf{v}}_{\rm tg}-{\mathbf n}^T\left(\frac{\partial {\mathbf v}_{\rm tg}}{\partial{\mathbf x}}\right)\mathbf{n}
={\rm div} {\mathbf{v}}_{\rm tg}+{\mathbf n}^T\left(\frac{\partial \mathbf{n}}{\partial{\mathbf x}}\right)^T{\mathbf v}_{\rm tg},
\end{equation*}
which proves relation \eqref{1}. To prove relation \eqref{2}, one uses the following identity valid for any vector fields $ {\mathbf a}$ and ${\mathbf b}$ :
\begin{equation*}
{\rm rot}\left({\mathbf a}\times{\mathbf b}\right)={\mathbf a}\,{\rm div{\mathbf b}}-{\mathbf b}\,{\rm div{\mathbf a}}+\frac{\partial {\mathbf a}}{\partial{\mathbf x}}\,{\mathbf b}-\frac{\partial {\mathbf b}}{\partial{\mathbf x}}\,{\mathbf a}. 
\end{equation*}
We apply  this identity to the vectors ${\mathbf a}={\mathbf n}$ and ${\mathbf b}={\mathbf v}_{\rm tg}$. Multiplying  on   left   by ${\mathbf n}^T$, one obtains relation \eqref{2}. Relations \eqref{3}, \eqref{4}, \eqref{5} are direct consequences of relation \eqref{1}. 
\section{Curvature tensor}
The unit normal vector being prolonged in the surface vicinity, we can directly obtain the expression of its derivative : 
\begin{equation*}
\frac{\partial \mathbf{n}}{\partial \mathbf{x}}=\mathbf{P}\; \frac{\varphi^{''}}{\vert\nabla\varphi\vert},
\end{equation*}
where $\varphi^{''}$ is the Hessian matrix of $\varphi$ with respect to $\mathbf{x}$. 
One obviously has 
\begin{equation*}
 \mathbf{n}^{T}\frac{\partial \mathbf{n}}{\partial\mathbf{x}}=\mathbf{0}.
\end{equation*}
However, since in dynamics $\mathbf{n}$ is not the gradient of the signed distance function, we cannot have the property : 
\begin{equation}
\frac{\partial \mathbf{n}}{\partial\mathbf{x}} \mathbf{n} =\mathbf{0}.
\label{hypothesis}
\end{equation}
The curvature tensor is defined as : 
\begin{equation*}
\mathbf{R} =
-  \mathbf{P}\; \frac{\varphi^{''}}{\vert\nabla\varphi\vert}\,\mathbf{P}= -\frac{\partial \mathbf{n}}{\partial \mathbf{x}}\,\mathbf{P}. 
\end{equation*}
Hence, in dynamics 
\begin{equation*}
 \mathbf{R}\neq -\frac{\partial \mathbf{n}}{\partial \mathbf{x}}.
\end{equation*}
Let us note that the derivation of the shape equation and boundary conditions in statics always uses property (\ref{hypothesis}) and  the curvature tensor coming from the definition of the signed distance function. In dynamics, we cannot use these properties and new tools should be developed.

Tensor $\mathbf{R}$ is symmetric and has zero as an eigenvalue :  
\begin{equation*}
\mathbf{R}=\mathbf{R}^T, \quad \mathbf{R}\mathbf{n}=\mathbf{0}.
\end{equation*}
In the  eigenbasis,   tensor   $\mathbf{R}$ is diagonal :
\begin{equation*}
\mathbf{R}=\left(
\begin{array}{ccc}c_1 & 0 & 0\\
0 & c_2 & 0 \\
0 & 0 & 0
 \end{array}
 \right),
\end{equation*}
where $c_1, c_2$ are the principal curvatures. The two invariants of   curvature tensor $\mathbf{R}$ are : 
\begin{equation*}
H=c_1+c_2, \quad K=c_1c_2. 
\end{equation*} 
Invariant $H$ is the double mean curvature,  and invariant $K$ is the Gaussian curvature. 
They can also be expressed in the form :
\begin{equation*}
H={\rm tr }\, \mathbf{R} =-{\rm tr }\left(\frac{\partial\mathbf{n}}{\partial \mathbf{x}}\right),
\end{equation*}
\begin{equation*}
2K=\left({\rm tr }\, \mathbf{R}\right) ^2-{\rm tr }\left(\mathbf{R}^2\right)=\left[{\rm tr }\left(\frac{\partial \mathbf{n}}{\partial \mathbf{x}}\right)\right]^2-{\rm tr }\left[\left(\frac{\partial \mathbf{n}}{\partial \mathbf{x}}\right)^2\right].
\label{H_K}
\end{equation*}
\begin{lemma}
\label{Lemma 1}
The following identities are valid : 
\begin{eqnarray}
&&{\rm div_{tg}} \,	\mathbf{P} = H\, 	\mathbf{n}^T,\notag\\
&&{\rm div_{tg}} \,	\mathbf{R} = {\rm \nabla_{tg}}^T\,H+\left(H^2-2\,K\right)\mathbf{n}^T,\notag\\
&&	\mathbf{R}^2 = H\, \mathbf{R}- K\, \mathbf{P}\notag.
\end{eqnarray}
\end{lemma}
\begin{proof}:  First, let us remark that $\mathbf{P}=\mathbf{P}_{\rm tg}$, $\mathbf{R}=\mathbf{R}_{\rm tg}$. One can apply Eq. \eqref{3} to obtain : 
\begin{eqnarray*}
 {\rm div_{tg}} \,	\mathbf{P}& =&-  {\rm div} \,	\left(\mathbf{n\,n}^T\right) +  {\mathbf n}^T\left(\frac{\partial {\mathbf n}}{\partial{\mathbf x}}\right)^T\mathbf{P}\\ &=& -  \left({\rm div} \,	\mathbf{n}\right)\mathbf{n} ^T -  {\mathbf n}^T\left(\frac{\partial {\mathbf n}}{\partial{\mathbf x}}\right)^T+  {\mathbf n}^T\left(\frac{\partial {\mathbf n}}{\partial{\mathbf x}}\right)^T\left(\mathbf{I-\mathbf{n\,n}^T}\right)\\
 &=& -  \left({\rm div} \,	\mathbf{n}\right)\mathbf{n} ^T,
	\end{eqnarray*}
 which proves the first relation. The proof of the second relation is as follows : 
\begin{eqnarray*}
	{\rm div} \,	\mathbf{R}& =&-{\rm div} \,\left(\frac{\partial {\mathbf n}}{\partial{\mathbf x}}\right)+ {\rm div} \,	\left(\frac{\partial {\mathbf n}}{\partial{\mathbf x}}\,\mathbf{n\,n}^T\right) \\ 
		&=&  -\frac{\partial({\rm div} \, {\mathbf n})}{\partial{\mathbf x}}+ {\rm div} \,\left(\frac{\partial {\mathbf n}}{\partial{\mathbf x}}\mathbf{n}\right)\,\mathbf{n}^T+ \mathbf{n}^T\left(\left(\frac{\partial {\mathbf n}}{\partial{\mathbf x}}\right)^{2}\right)^T\\
	&=&  -\frac{\partial{(\rm div} \, {\mathbf n})}{\partial{\mathbf x}}+ {\rm div} \,\left(\frac{\partial {\mathbf n}}{\partial{\mathbf x}}\right)\,\mathbf{n\,n}^T+{\rm tr}\left(\left(\frac{\partial {\mathbf n}}{\partial{\mathbf x}}\right)^2\right)\mathbf{n}^T+ \mathbf{n}^T\left(\left(\frac{\partial {\mathbf n}}{\partial{\mathbf x}}\right)^{2}\right)^T\\
	&=& \frac{\partial H}{\partial\mathbf {x}}\,\mathbf {P}+ {\rm tr}\left(\left(\frac{\partial {\mathbf n}}{\partial{\mathbf x}}\right)^2\right)\mathbf{n}^T- \mathbf{n}^T\left(\frac{\partial {\mathbf n}}{\partial{\mathbf x}}\right)^T\mathbf{R}.
\end{eqnarray*}
Consequently,
\begin{equation*}
{\rm div_{tg}} \,	\mathbf{R}= \frac{\partial H}{\partial\mathbf {x}}\,\mathbf {P}+ {\rm tr}\left(\left(\frac{\partial {\mathbf n}}{\partial{\mathbf x}}\right)^2\right)\mathbf{n}^T.
\end{equation*}
Using  $\displaystyle  {\rm tr}\left(\left(\frac{\partial {\mathbf n}}{\partial{\mathbf x}}\right)^2\right)= {\rm tr}\left({\mathbf R}^2\right) = H^2-2\,K$, we obtain the second relation of the lemma.  
\\
Now, the curvature tensor satisfies the Cayley-Hamilton theorem : 
\begin{equation*}
{\mathbf R}^3-H \, {\mathbf R}^2+ K\, {\mathbf R} =0.
\end{equation*}
The minimal polynomial is :
\begin{equation*}
{\mathbf R}^2-H \, {\mathbf R} + K\, {\mathbf P} =0,\label{minimal}
\end{equation*}
which proves the third relation.
	\end{proof}
\section{Virtual motion}
Let a one-parameter family of 
\emph{virtual motions} 
\begin{equation*}
\mathbf{x}=\mathbf{\Phi}\left(t, \mathbf{X},  \lambda \right) 
\label{vitual motion}
\end{equation*}
with scalar $\lambda \in O$, where $O$ is an open real interval containing zero and
such that $\mathbf{\Phi}\left( t, \mathbf{X}, 0\right) =\boldsymbol{\phi }%
\left(t,  \mathbf{X}\right) $ (the  motion of the continuous medium is obtained for $\lambda =0$). The {\it virtual displacement} of particle $\mathbf X$  is defined as   \cite{Gavrilyuk,Serrin} :  
\begin{equation*}
\delta\mathbf{x}(t, \mathbf{X})=\frac{\partial {\mathbf{\Phi}}(t,\mathbf{X}, \lambda)}{\partial \lambda}\vert_{\lambda=0}. 
\label{vitual displacement}
\end{equation*}
In the following, symbol $\delta$ means the derivative with respect to $\lambda$  at fixed Lagrangian coordinates $\mathbf{X}$ and $t$, for $\lambda=0$. 
We will also denote by $\boldsymbol{\zeta}(t,{\mathbf x})$ the virtual displacement expressed as a function of Eulerian coordinates :
\begin{equation*}
\boldsymbol{\zeta}(t,\mathbf x)=\boldsymbol{\zeta}\left(t,{\boldsymbol \phi}(t, \mathbf X)\right)=\delta{\mathbf x}\left(t, \mathbf X\right). 
\end{equation*}
\section{Variational  tools}
 We assume that   $D_t$  has a smooth boundary $S_t$ with  edge $C_t$. We respectively denote $D_0$,  $S_0$ and $C_0$ the images of $D_t$, $S_t$ and $C_t$ in the reference space (of Lagrangian coordinates). The unit vector $\mathbf{n}$ and its image $\mathbf{n}_0$ are the oriented normal vectors to $S_t$ and
$S_0$; the vector $\mathbf{t}$ is  the oriented unit tangent vector to $C_t$ and $\mathbf{n}^{\prime}= \mathbf{t}\times \mathbf{n}$ is the unit binormal vector (see Fig. 1).
$\mathbf{F} =\partial {{\boldsymbol{\phi}}(t, \mathbf X)}/\partial \mathbf{X}\equiv \partial \mathbf{x}/\partial \mathbf{X}$ is the deformation gradient. For the sake  of simplicity, we will use  the same notations for quantities as  $\mathbf{F}$, $\mathbf n$, etc. both  in Eulerian and Lagrangian coordinates.  
\begin{lemma}
\label{basic_lemma}
We have the  relations :
\begin{eqnarray}
\delta \det \mathbf{F} &=& \det \mathbf{F}\, \rm{div}\,\boldsymbol{\zeta}\,,  \label{Jacobi} \\
\delta \mathbf{n} &=& -\mathbf{P}\,\left(\frac{\partial \boldsymbol{\zeta}}{\partial \mathbf{x}}\right)^T\mathbf{n} \label{variation_n}\,, \\
\delta \left(\mathbf{F}^{-1}\mathbf{n}\right) &=&-\mathbf{F}^{-1}\ \frac{\partial \boldsymbol{\zeta }}{\partial \mathbf{x}}\,\mathbf{n}+\mathbf{F}^{-1} \delta \mathbf{n}\,,  \label{InverseF}\\
\delta\left(\frac{\partial \mathbf{n}}{\partial \mathbf{x}}\right)&=&\frac{\partial \delta\mathbf{n}}{\partial \mathbf{x}}-{\frac{\partial \mathbf{n}}{\partial \mathbf{x}}\frac{\partial \boldsymbol{\zeta}}{\partial \mathbf{x}}}\,.\label{variation_gradientn}
\end{eqnarray}
\end{lemma}
\emph{Proof of} Rel. (\ref{Jacobi}): \newline
The   Jacobi formula for determinant  is  :
\begin{equation*}
\delta (\det \mathbf{F})=\det \mathbf{F}\,{\rm tr}\left( \mathbf{F}^{-1}\delta \mathbf{F}\right).
\end{equation*}
Also, 
\begin{equation*}
\delta \mathbf{F}=\delta \left( \frac{\partial \mathbf{x}}{\partial \mathbf{X}}\right) =  
\frac{\partial \delta{\mathbf x}}{\partial \mathbf{X}}.
\end{equation*}%
Then
\begin{equation*}
{\rm tr}\left( \mathbf{F}^{-1}\delta \mathbf{F}\right) ={\rm tr}\left( \frac{\partial
\mathbf{X}}{\partial \mathbf{x}}\,\frac{\partial \delta{\mathbf x}}{\partial
\mathbf{X}}\right) ={\rm tr}\left( \frac{\partial \delta{\mathbf x}}{\partial \mathbf{X}}\,\frac{\partial \mathbf{X}}{\partial \mathbf{x}}\right)
={\rm tr}\left( \frac{\partial \boldsymbol{\zeta}}{\partial \mathbf{x}}%
\right) =\mathrm{div}\,{\boldsymbol{\zeta}}.
\end{equation*}
\emph{Proof of} Rel. (\ref{variation_n}): \newline
Surface $\varphi(t, \mathbf{x})=0$ is a material surface. It can be represented in the Lagrangian coordinates as $\varphi(t, \mathbf{x}) =\varphi_0 (\mathbf{X})$ which implies that $\delta \varphi=0$. Also, 
\begin{equation*}
\delta \left( \frac{\partial \varphi}{\partial \mathbf{x}}\right)=\delta \left( \frac{\partial \varphi}{\partial \mathbf{X}}\mathbf{F}^{-1}\right)=\frac{\partial \delta\varphi}{\partial \mathbf{x}}-\frac{\partial \varphi}{\partial \mathbf{x}}\frac{\partial \boldsymbol{\zeta}}{\partial \mathbf{x}}=-\frac{\partial \varphi}{\partial \mathbf{x}}\frac{\partial \boldsymbol{\zeta}}{\partial \mathbf{x}}.
\end{equation*}
Here  we used the following expression for the variation of $\mathbf{F}^{-1}$ coming from the relation $\mathbf{F}^{-1}\mathbf{F}=\mathbf{I}$ :
\begin{equation*}
\delta \mathbf{F}^{-1}=-\mathbf{F}^{-1}\frac{\partial \boldsymbol{\zeta}}{\partial \mathbf{x}}.
\end{equation*}
One also has  :
\begin{equation*}
\delta \vert\nabla\varphi\vert=\frac{\left(\nabla\varphi\right)^T\delta\nabla\varphi}{\vert\nabla\varphi\vert}.
\end{equation*}
Finally, taking the variation of\ \, 
$\displaystyle 
\mathbf{n}=\frac{\nabla\varphi}{\vert\nabla\varphi\vert},
$
one can obtain 
\begin{equation*}
\delta \mathbf{n}=\left(\mathbf{n}^T\mathbf{n}-\mathbf{I}\right)\left(\frac{\partial \boldsymbol{\zeta}}{\partial \mathbf{x}}\right)^T\mathbf{n}=-\mathbf{P}\left(\frac{\partial \boldsymbol{\zeta}}{\partial \mathbf{x}}\right)^T\mathbf{n}.
\end{equation*}
\emph{Proof of} Rel. (\ref{InverseF}):  \newline
\begin{equation*}
\delta \left( \mathbf{F}^{-1}\mathbf{n}\right) =\delta \left( \mathbf{F}^{-1}\right) \mathbf{n}
+\mathbf{F}^{-1}\delta \mathbf{n}=-\mathbf{F}^{-1}\ \frac{\partial \boldsymbol{\zeta }}{\partial \mathbf{x}}\,\mathbf{n}+\mathbf{F}^{-1} \delta \mathbf{n}.
\end{equation*}
\emph{Proof of} Rel. (\ref{variation_gradientn}):  
\begin{equation*}
\delta\left(\frac{\partial \mathbf{n}}{\partial \mathbf{x}}\right)=\delta\left(\frac{\partial \mathbf{n}}{\partial \mathbf{X}}\mathbf{F}^{-1}\right)=\frac{\partial \delta\mathbf{n}}{\partial \mathbf{X}}\mathbf{F}^{-1}+\frac{\partial \mathbf{n}}{\partial \mathbf{X}}\delta\mathbf{F}^{-1}=\frac{\partial \delta\mathbf{n}}{\partial \mathbf{x}}-{\frac{\partial \mathbf{n}}{\partial \mathbf{x}}\frac{\partial \boldsymbol{\zeta}}{\partial \mathbf{x}}}.
\end{equation*}
\\
We denote by $\sigma$ the energy per unit area of surface $S_t$. The variation of $\sigma$ is $\delta\sigma$. This variation depends on the  physical problem through the dependence of $\sigma$ on geometrical and thermodynamical parameters. For now, we do not need to know this variation in explicit form,  the variation will be given further. The next lemma gives  the variation of the surface potential energy \cite{Gouin_2014,Gouin_2014_vesicle}.  
\begin{lemma}
\label{variation_surface_energy}
 Let us consider a material surface  $S_t$ of  boundary edge $C_t$.  The variation of surface energy
 \begin{equation*}
 E=\protect\iint_{S_t} \sigma\,ds
 \end{equation*}
  is   
\begin{equation*}
\delta E=\iint_{S_t}\left[\delta \sigma-\left(  \boldsymbol{\nabla}_{\rm tg}^{T}\sigma
 +\sigma H \,\mathbf{n}^{T}\right)\boldsymbol{\zeta}\right] ds + \int_{C_t} \sigma\,
\mathbf{n^{\prime}}^T \boldsymbol{\zeta}~dl,  \label{varsurfO}
\end{equation*}
where $ds$, $dl$ are the surface and  line measures, respectively\,\footnote{It is interesting to remark that the combination $\displaystyle{\hat\delta \sigma=\delta \sigma- \left(\boldsymbol{\nabla}_{\rm tg}^{T}\sigma\right) \boldsymbol{\zeta}  }$ is the variation of $\sigma$ at fixed Eulerian coordinates. Indeed,  since the symbol $\delta$ means the variation at fixed Lagrangian coordinates, and $\hat\delta$ is the variation at fixed Eulerian coordinates, this formula is a natural general relation between two types of variations (cf. \cite{SGHG,Gavrilyuk}).}.
\end{lemma}

\begin{proof}:
We suppose that the unit normal vector field is locally extended
in the vicinity of $S_t$. For any vector field ${\mathbf{w}}$ one has :
\begin{equation*}
\mathrm{rot} ({\mathbf{n}} \times {\mathbf{w}}) ={\mathbf{n }}\,
\mathrm{div} \, {\mathbf{w}}-
{\mathbf{w}}\,\mathrm{div}\,{\mathbf{n }} + \frac {\partial
{\mathbf{n}}} {\partial {\mathbf{x}}}\, {\mathbf{w}}- \frac {\partial {%
\mathbf{w}}} {\partial {\mathbf{x}}}\, {\mathbf{n}}.
\end{equation*}
From relation ${\mathbf{n}}^T{\mathbf{n}}=1$, we obtain  $\displaystyle {\mathbf{n}}^T\frac {\partial {\mathbf{n}}} {\partial {\mathbf{x}}}= 0$. Using the definition of $H$,\, ($ 
 \displaystyle  H= - \mathrm{div}\,{\mathbf{n }}$),  we deduce on $S_t$ :
\begin{equation}
{\mathbf{n}^T} \mathrm{rot} ({\mathbf{n}} \times {\mathbf{w}}) =
\mathrm{div}
\, {\mathbf{w}}+H\;  {\mathbf{n}^T} {\mathbf{w}} - {\mathbf{n}^T%
} \frac {\partial {\mathbf{w}}} {\partial {\mathbf{x}}}\,
{\mathbf{n}}. \label{A0}
\end{equation}
{The surface energy is given by :
\begin{equation*}
E =\iint_{S_t} \sigma \; \vert d_1{\mathbf{x}}\wedge d_2{\mathbf{x}}\vert,
\end{equation*}
where $\displaystyle{d_i{\mathbf{x}}=\frac{\partial \mathbf{x}}{\partial s_i}ds_i} \; (i=1,2) $  and  $s_i$ are curvilinear coordinates on $S_t$. This integral can also be written as :  
\begin{equation*}
E =\iint_{S_t} \sigma\ \det\, ({\mathbf{n}},d_1{\mathbf{x}},d_2{\mathbf{x}})=\iint_{S_0}\sigma\, \hbox{det} (\mathbf{F}\mathbf{F}^{-1}{\mathbf{n}},\mathbf{F}d_{10}{\mathbf{X}},\mathbf{F}d_{20}{\mathbf{X}}).
\end{equation*}
Here $\displaystyle{d_{i0}\mathbf{X}=\frac{\partial \mathbf{X}}{\partial s_{i0}}ds_{i0}}$ and  $s_{i0}$ are the corresponding curvilinear coordinates on $S_0$.
Finally, 
\begin{equation*}
E =\iint_{S_0}\sigma\;   (\hbox{det}\mathbf{F}) \;  \hbox{det} (\mathbf{F}^{-1}{\mathbf{n}}, d_{10}{\mathbf{X}},d_{20}{\mathbf{X}})=\iint_{S_0}\sigma\;     \hbox{det} ((\hbox{det}\mathbf{F}) \;\mathbf{F}^{-1}{\mathbf{n}}, d_{10}{\mathbf{X}},d_{20}{\mathbf{X}}).
\end{equation*}
Let us remark that $\displaystyle{(\hbox{det}\mathbf{F}) \;\mathbf{F}^{-1}{\mathbf{n}}}$  is the image of $\mathbf{n}$ and is not the normal vector   to $S_0$ because $\mathbf{F}$ is not an orthogonal transformation.} 

One has :
\begin{equation*}
\begin{array}{cc}
\displaystyle \delta E =\iint_{S_0}\delta \sigma\ \det \mathbf{F}\
\hbox{det}\,
(\mathbf{F}^{-1}{\mathbf{n}},d_{10}{\mathbf{X}},d_{20}{\mathbf{X}}) \\
\displaystyle +
\iint_{S_0}\sigma\, \delta\big(\det  \mathbf{F}\ \hbox{det}\, (\mathbf{F}^{-1}{\mathbf{n}%
},d_{10}{\mathbf{X}},d_{20}{\mathbf{X}})\big).
\end{array}%
\end{equation*}
Using Lemma \ref{basic_lemma}, one gets :
\begin{equation*}
\begin{array}{cc}
\displaystyle \iint_{S_0}\sigma\, \delta\big(\det \mathbf{F}\
\hbox{det}\,
(\mathbf{F}^{-1}{\mathbf{n}},d_{10}{\mathbf{X}},d_{20}{\mathbf{X}})\big) = &  \\
\displaystyle \iint_{S_t}  \sigma\ \mathrm{div}\, {\boldsymbol{\zeta}} \ \det ({%
\mathbf{n}},d_1{\mathbf{x}},d_2{\mathbf{x}}) + \sigma\, \det \left (%
\displaystyle  \delta \mathbf{n},d_1{\mathbf{x}},d_2{\mathbf{x}}\right )  & 
\displaystyle -
\sigma \det \left (\displaystyle \frac {\partial{\boldsymbol{\zeta}}}{\partial {\mathbf{x}}}\,{\mathbf{n}}, d_1{\mathbf{x}}, d_2{\mathbf{x}} \right )     \\
= \displaystyle \iint_{S_t} \left( \mathrm{div} (\sigma\,{\boldsymbol{\zeta}} )-(%
\boldsymbol{\nabla}^{T} \sigma) \, {\boldsymbol{\zeta}} -\sigma {\mathbf{n}}^T \frac {%
\partial {\boldsymbol{\zeta}}} {\partial {\mathbf{x}}} \, {\mathbf{n}} \right )
ds. &
\end{array}%
\end{equation*}
Relation (\ref{A0}) yields
\begin{equation*}
\displaystyle \mathrm{div}\, (\sigma\,{\boldsymbol{\zeta}}) +   \sigma H\, {\mathbf{n}}^T {\boldsymbol{\zeta}} - {\mathbf{n}}^T \frac
{\partial
(\sigma\, {\boldsymbol{\zeta}})} {\partial {\mathbf{x}}} \, {\mathbf{n}} = {%
\mathbf{n}}^T\, \mathrm{rot}\, (\sigma\,{\mathbf{n}}\times
{\boldsymbol{\zeta}} ).
\end{equation*}
It implies 
\begin{equation}
\begin{array}{cc}
\displaystyle \iint_{S_0}\sigma\,\delta \big(\det
\mathbf{F}\,\hbox{det}\,
(\mathbf{F}^{-1}{\mathbf{n}},d_{10}{\mathbf{X}},d_{20}{\mathbf{X}})\big) = & \label{variationE} \\
\displaystyle\ \iint_{S_t} - \left(\sigma  {H} \, \mathbf{n}^T+
(\boldsymbol{\nabla}^{T}\sigma)\, \mathbf{P}   \right)
{\boldsymbol{\zeta}}
\, ds+\iint_{S_t}{\mathbf{n}}^T\ \mathrm{rot}\, (\sigma\,{\mathbf{n}}\times
\boldsymbol{\zeta})\,ds. \notag &
\end{array}%
\end{equation}
Since $\mathbf{P}\,  \boldsymbol{\nabla}\sigma    \equiv
 \boldsymbol{\nabla}_{\rm tg}\sigma$, one has  
\begin{equation*}
\iint_{S_t}{\mathbf{n}}^T\ \mathrm{rot}\, (\sigma\,{\mathbf{n}}\times{%
\boldsymbol{\zeta}})\,ds
=\int_{C_t}{\rm det}(\mathbf{t},\sigma\,{\mathbf{n}},\boldsymbol{\zeta})\,dl
= \int_{C_t} \sigma\, \mathbf{n^{\prime}}^T \boldsymbol{\zeta}~dl,
\end{equation*}
and we  obtain  Lemma \ref{varsurfO}.
\end{proof}
\begin{lemma}
\label{derivative_with_respect_to_R}
Let $\sigma$ be a function of curvature tensor $\mathbf{R}$, or equivalently, a function of $H$ and $K$. Then,
\begin{equation}
\frac{\partial\sigma}{\partial\mathbf{R}}=
a\;\mathbf{I}+b\; \mathbf{R}\quad {\rm with}\quad a=
\frac{\partial\sigma}{\partial H}+ 
H\frac{\partial\sigma}{\partial K}\quad {\rm and}\quad
b=-\frac{\partial\sigma}{\partial K}, 
\label{lemmea}
\end{equation}
where for the sake of simplicity, we indifferently write $\sigma(\mathbf{R})$ or $\sigma(H,K)$. In particular, this implies :
\begin{equation}
 \mathbf{n}^T
 \frac{\partial\sigma}{\partial\mathbf{R}}\;\frac{\partial\mathbf{n}}{\partial\mathbf{x}}=\mathbf{0}.
 \label{lemmeb}
\end{equation}
\end{lemma}
\begin{proof}:
Since $H={\rm tr}\,\mathbf{R}$, $2K=\left({\rm tr}\,\mathbf{R}\right)^2-{\rm tr}\left(\mathbf{R}^2\right)$, and 
\begin{equation*}
\frac{\partial {\rm tr}\left(\mathbf{R}^k\right)}{\partial \mathbf{R}}=k\;\mathbf{R}^{k-1},
\end{equation*}
one gets 
\begin{equation*}
\frac{\partial\sigma}{\partial\mathbf{R}}=\left(\frac{\partial\sigma}{\partial H}+ 
H\frac{\partial\sigma}{\partial K}\right)\mathbf{I}-\frac{\partial\sigma}{\partial K}\mathbf{R}.
\end{equation*}
Since
\begin{equation}
\mathbf{R} =-\frac{\partial\mathbf{n}}{\partial\mathbf{x}}\,
\mathbf{P}\quad {\rm and}\quad
\frac{\partial\sigma}{\partial\mathbf{R}}= a \;\mathbf{I}+b
\mathbf{R},\label{R}
\end{equation}
we obtain 
\begin{equation*}
\mathbf{n}^T
 \frac{\partial\sigma}{\partial\mathbf{R}}\;\frac{\partial\mathbf{n}}{\partial\mathbf{x}}=
 a\; \mathbf{n}^T\frac{\partial\mathbf{n}}{\partial\mathbf{x}}-b\;\mathbf{n}^T
 \left(\frac{\partial\mathbf{n}}{\partial\mathbf{x}}\right)^2=\mathbf{0}.
\end{equation*}
\end{proof}

\section{Variation of $\sigma$}
This is a key part of the paper. The variation of the surface energy per unit area is obtained   in  the general case $\sigma =\sigma (H,K)$.  
The membrane is determined by a surface $S_t$ having  a closed contact  line  $C_t$ on a rigid surface $\mathcal{S}=S_1\cup S_2$ (see Fig. \ref{fig1}).  The dependence on other parameters such as    concentrations of surfactants on the membranes can further  be taken into account as in  \cite{Gouin_2014,Steigmann1}.
\begin{lemma}\label{key1}
The variation of surface energy $\sigma (\mathbf{R})$ is given by the relation :
\begin{equation}
\delta \sigma =  -{\rm div_{tg}} \,\left(\frac{\partial\sigma}{\partial\mathbf{R}}\mathbf{R}\, \boldsymbol{\zeta} +	\mathbf{P}\frac{\partial\sigma}{\partial\mathbf{R}}\,\delta\mathbf{n}\right)
+{\rm div_{tg}}\left(\frac{\partial\sigma}{\partial\mathbf{R}}\mathbf{R}\right)  \boldsymbol{\zeta}+{\rm div_{tg}}\left(\mathbf{P}\frac{\partial\sigma}{\partial\mathbf{R}} \right)  \delta\mathbf{n}. \label{sigmavar}
\end{equation}
\end{lemma}
\begin{proof}:
	Using Lemma \ref{basic_lemma},  we have :
	\begin{equation*}
\delta  \mathbf{R}= - \delta\left(\frac{\partial\mathbf{n} }{\partial \mathbf{x}}\, \mathbf{P}\right) =-\left( \frac{\partial\delta\mathbf{n} }{\partial \mathbf{x}}-\frac{\partial\mathbf{n} }{\partial \mathbf{x}}\frac{\partial\boldsymbol{\zeta} }{\partial \mathbf{x}}\right)  \mathbf{P}+\frac{\partial\mathbf{n} }{\partial \mathbf{x}}\,\delta\left( \mathbf{n}\mathbf{n}^T\right). 
	\end{equation*}
	By taking account of  Eq. \eqref{variation_n} and $ \delta\left( \mathbf{n}\mathbf{n}^T\right) =  \delta  \mathbf{n} \, \mathbf{n}^T+  \mathbf{n}\,\delta\mathbf{n}^T,$ we get :
		\begin{equation*}
	\delta  \mathbf{R}= -\frac{\partial\delta\mathbf{n} }{\partial \mathbf{x}}\, \mathbf{P}+\frac{\partial\mathbf{n} }{\partial \mathbf{x}}\frac{\partial\boldsymbol{\zeta} }{\partial \mathbf{x}}\, \mathbf{P} -\frac{\partial\mathbf{n} }{\partial \mathbf{x}}\, \mathbf{P}\,\left(\frac{\partial\boldsymbol{\zeta} }{\partial \mathbf{x}}\right) ^T \mathbf{n}\mathbf{n}^T  -\frac{\partial\mathbf{n} }{\partial \mathbf{x}}\, \mathbf{n}\mathbf{n}^T\frac{\partial\boldsymbol{\zeta} }{\partial \mathbf{x}}\, \mathbf{P}.
	\end{equation*}
	We deduce :
\begin{eqnarray*}
	\delta\sigma &= & {\rm tr }\left(\frac{\partial\sigma}{\partial\mathbf{R}}\delta\mathbf{R}\right)\\
	&=& {\rm tr }\left[\frac{\partial\sigma}{\partial\mathbf{R}}\left(-\frac{\partial\delta\mathbf{n} }{\partial \mathbf{x}} \mathbf{P}+\frac{\partial\mathbf{n} }{\partial \mathbf{x}}\frac{\partial\boldsymbol{\zeta} }{\partial \mathbf{x}} \mathbf{P} -\frac{\partial\mathbf{n} }{\partial \mathbf{x}}\, \mathbf{P}\,\left(\frac{\partial\boldsymbol{\zeta} }{\partial \mathbf{x}}\right) ^T \mathbf{n}\mathbf{n}^T  -\frac{\partial\mathbf{n} }{\partial \mathbf{x}}\, \mathbf{n}\mathbf{n}^T\frac{\partial\boldsymbol{\zeta} }{\partial \mathbf{x}}\, \mathbf{P}\right) \right].
\end{eqnarray*} 
From Eq.  \eqref{lemmeb}, we get $\displaystyle \mathbf{n}\mathbf{n}^T\frac{\partial\sigma}{\partial\mathbf{R}}
\frac{\partial\mathbf{n} }{\partial \mathbf{x}}\, \frac{\partial\boldsymbol{\zeta} }{\partial \mathbf{x}}=0 $ and  $\displaystyle \mathbf{n}\mathbf{n}^T\frac{\partial\sigma}{\partial\mathbf{R}}
\frac{\partial\mathbf{n} }{\partial \mathbf{x}}\, \mathbf{n}\mathbf{n}^T \frac{\partial\boldsymbol{\zeta} }{\partial \mathbf{x}}=0 $. \\ 

\noindent Consequently, $\displaystyle \frac{\partial {\sigma} }{\partial \mathbf{R}}\, \frac{\partial\mathbf{n} }{\partial \mathbf{x}}\,\mathbf{P} \, \frac{\partial\boldsymbol{\zeta} }{\partial \mathbf{x}}   = -\frac{\partial {\sigma} }{\partial \mathbf{R}}\, \mathbf{R}\, \frac{\partial\boldsymbol{\zeta} }{\partial \mathbf{x}}  $, which implies :
\begin{eqnarray*}
	\delta\sigma  
	&=& -{\rm tr }\left[ \mathbf{P}\frac{\partial\sigma}{\partial\mathbf{R}} \frac{\partial\delta\mathbf{n} }{\partial \mathbf{x}} +\frac{\partial {\sigma} }{\partial \mathbf{R}}\, \mathbf{R}\, \frac{\partial\boldsymbol{\zeta} }{\partial \mathbf{x}}       \right]\\
	&=& -{\rm div\left( \mathbf{P}\frac{\partial\sigma}{\partial\mathbf{R}} \delta  \mathbf{n}\right) }+{\rm div\left( \mathbf{P}\frac{\partial\sigma}{\partial\mathbf{R}} \right) }\delta  \mathbf{n}-{\rm div\left(\frac{\partial\sigma}{\partial\mathbf{R}}  \mathbf{R}\,\boldsymbol{\zeta}\right) }+{\rm div\left( \frac{\partial\sigma}{\partial\mathbf{R}}\,\mathbf{R} \right) }\boldsymbol{\zeta}.
\end{eqnarray*} 
By taking account of Eq. \eqref{1}, we get :
\begin{eqnarray*}
\delta\sigma  
&=& -{\rm div_{tg}\left( \mathbf{P}\frac{\partial\sigma}{\partial\mathbf{R}} \delta  \mathbf{n}\right) }+{\rm div_{tg}\left( \mathbf{P}\frac{\partial\sigma}{\partial\mathbf{R}} \right) }\delta  \mathbf{n}-
{\rm div_{tg}}\left(\frac{\partial\sigma}{\partial\mathbf{R}}\,\mathbf{R}  \boldsymbol{\zeta}\right) +{\rm div_{tg}}\left(\frac{\partial\sigma}{\partial\mathbf{R}}\,\mathbf{R}\right) \boldsymbol{\zeta} ,
\end{eqnarray*} 
and relation (\ref{sigmavar}) is proven. 

\end{proof}
Now, we have to study term $\displaystyle{{\rm div_{tg}\left( \mathbf{P}\frac{\partial\sigma}{\partial\mathbf{R}} \right) }\delta  
	\mathbf{n}}$. 
\begin{lemma}\label{key2}
\begin{eqnarray*}
	{\rm div_{tg}\left( \mathbf{P}\,\frac{\partial\sigma}{\partial\mathbf{R}} \right) }\delta  \mathbf{n} &= & -\,{\rm div_{tg}}\left[\mathbf{P}\,{\rm div_{tg}}^T\left(\mathbf{P}\,\frac{\partial\sigma}{\partial\mathbf{R}}\right)\mathbf{n}^T\boldsymbol{\zeta}\right] \\
	&+& {\rm div_{tg}}\left[\mathbf{P}\,{\rm div_{tg}}^T\left(\mathbf{P}\,\frac{\partial\sigma}{\partial\mathbf{R}}\right)\right] \mathbf{n}^T\boldsymbol{\zeta}   -\, {\rm div_{tg}}\left(\mathbf{P}\,\frac{\partial\sigma}{\partial\mathbf{R}}\right)\mathbf{R}\,\boldsymbol{\zeta} .
\end{eqnarray*} 
\end{lemma}
\begin{proof}:
Using  relation \eqref{variation_n}, one  obtains :
\begin{eqnarray*}
	{\rm div_{tg}\left( \mathbf{P}\frac{\partial\sigma}{\partial\mathbf{R}} \right) }\delta  \mathbf{n} &= & -	{\rm div}_{\rm tg}\left(\mathbf{P}\,\frac{\partial\sigma}{\partial\mathbf{R}}\right)\mathbf{P}\,\left(\frac{\partial\boldsymbol{\zeta} }{\partial \mathbf{x}}\right) ^T\mathbf{n}\\ &=& -	{\rm div_{\rm tg}}\left(\mathbf{P}\,\frac{\partial\sigma}{\partial\mathbf{R}}\right)\mathbf{P}\left[\left(\frac
	{\partial\left(\mathbf{n}^T\boldsymbol{\zeta}\right)}{\partial\mathbf{x}}\right)^T-\left(\frac{\partial\mathbf{n} }{\partial \mathbf{x}}\right)^T\boldsymbol{\zeta}\right]\\
	&=&-{\rm div_{tg}\left( \mathbf{P}\frac{\partial\sigma}{\partial\mathbf{R}} \right) }\nabla_{\rm tg}\left(\mathbf{n}^T\boldsymbol{\zeta}\right)-{\rm div_{tg}\left( \mathbf{P}\frac{\partial\sigma}{\partial\mathbf{R}} \right) }\mathbf{R}\,\boldsymbol{\zeta}\\
	&=& {\rm div_{tg}}\left[\mathbf{P}{\rm div_{tg}}^T\left(\mathbf{P}\frac{\partial\sigma}{\partial\mathbf{R}}\right)\right] \mathbf{n}^T\boldsymbol{\zeta}   -\, {\rm div_{tg}}\left(\mathbf{P}\frac{\partial\sigma}{\partial\mathbf{R}}\right)\mathbf{R}\boldsymbol{\zeta} \\
	&- & {\rm div_{tg}}\left[\mathbf{P}{\rm div_{tg}}^T\left(\mathbf{P}\frac{\partial\sigma}{\partial\mathbf{R}}\right)\mathbf{n}^T\boldsymbol{\zeta}\right].
	\end{eqnarray*} 
\end{proof}
Now, from Lemma \ref{varsurfO} and formula \eqref{sigmavar}, we obtain the following fundamental lemma. 
\begin{lemma}
	\label{key3}
	The variation of  surface energy $\displaystyle E=\iint_{S_t}\sigma\, ds $,  where $S_t$ has an oriented  boundary line $C_t$ with tangent unit vector $\mathbf{t}$ and binormal unit vector $\mathbf{n}'=  \mathbf{t} \times  \mathbf{n}$,  is given by the relation :
	\begin{eqnarray*}
	\delta E &=&  \iint_{S_t}\Big [\,{\rm div_{tg}} \left(\frac{\partial\sigma}{\partial\mathbf{R}}\mathbf{R}\right)  -{\rm div_{tg}}\left(	\mathbf{P}\frac{\partial\sigma}{\partial\mathbf{R}}\right)\mathbf{R}
	+{\rm div_{tg}}\left(\mathbf{P}\,{\rm div_{tg}^{\it T}} \left(	\mathbf{P}\frac{\partial\sigma}{\partial\mathbf{R}}\right)\right)  \mathbf{n}^T\notag   
	\\
&& \qquad -\,  \sigma\,H \,\mathbf{n}^T-\nabla_{\rm tg}^T\, \sigma\,\Big]{\boldsymbol \zeta}\; ds\label{Evar}\\
&+&  \int_{C_t}\mathbf{n}^{\prime T} \left\{ \left[\sigma\mathbf{I}-  	\frac{\partial\sigma}{\partial\mathbf{R}}\mathbf{R}-{\rm div_{tg}^{\it T}} \left(	\mathbf{P}\frac{\partial\sigma}{\partial\mathbf{R}}\right){\mathbf{n}}^T\right]\boldsymbol{\zeta} +\frac{\partial\sigma}{\partial\mathbf{R}}\mathbf{P}\left(\frac{\partial \boldsymbol{\zeta}}{\partial \mathbf{x}}\right)^T\mathbf{n} \right\}dl.\notag
	\end{eqnarray*}
\end{lemma}
\medskip

 \begin{proof} :
 By taking account of	Lemma \ref{key1} and Lemma \ref{key2}, we get
 \begin{eqnarray*}
 \delta \sigma &= & -{\rm div_{tg}} \,\left[\frac{\partial\sigma}{\partial\mathbf{R}}\mathbf{R}\, \boldsymbol{\zeta} +	\mathbf{P}\frac{\partial\sigma}{\partial\mathbf{R}}\,\delta\mathbf{n}+\mathbf{P}\,{\rm div_{tg}}^T\left(\mathbf{P}\frac{\partial\sigma}{\partial\mathbf{R}}\right)\mathbf{n}^T\boldsymbol{\zeta}\right]\\
 &+&\left[{\rm div_{tg}} \left(	\frac{\partial\sigma}{\partial\mathbf{R}}\mathbf{R}\right) -\, {\rm div_{tg}}\left(\mathbf{P}\frac{\partial\sigma}{\partial\mathbf{R}}\right)\mathbf{R} +{\rm div_{tg}}\left(\mathbf{P}{\rm div_{tg}}^T\left(\mathbf{P}\frac{\partial\sigma}{\partial\mathbf{R}}\right) \right)\mathbf{n}^T\right]\boldsymbol{\zeta} .
 \end{eqnarray*}
By using Eq. \eqref{2} and Lemma \ref{variation_surface_energy} associated with the Stokes formula, and   property  $\mathbf{n}^{\prime T} \mathbf{P}=\mathbf{n}^{\prime T}$, we obtain :
\begin{eqnarray*}
\delta E &=&  \iint_{S_t}\Big [\,{\rm div_{tg}} \left(\frac{\partial\sigma}{\partial\mathbf{R}}\mathbf{R}\right)  -{\rm div_{tg}}\left(	\mathbf{P}\frac{\partial\sigma}{\partial\mathbf{R}}\right)\mathbf{R}
+{\rm div_{tg}}\left(\mathbf{P}\,{\rm div_{tg}^{\it T}} \left(	\mathbf{P}\frac{\partial\sigma}{\partial\mathbf{R}}\right)\right)  \mathbf{n}^T\notag   
\\
&& \qquad -\,  \sigma\,H \,\mathbf{n}^T-\nabla_{\rm tg}^T\, \sigma\,\Big]{\boldsymbol \zeta}\; ds\label{Evar}\\
&+&  \int_{C_t}\mathbf{n}^{\prime T} \left\{ \left[\sigma\mathbf{I}-  	\frac{\partial\sigma}{\partial\mathbf{R}}\mathbf{R}-{\rm div_{tg}^{\it T}} \left(	\mathbf{P}\frac{\partial\sigma}{\partial\mathbf{R}}\right){\mathbf{n}}^T\right]\boldsymbol{\zeta} -\frac{\partial\sigma}{\partial\mathbf{R}}\,\delta\mathbf{n} \right\}dl.\notag
\end{eqnarray*}
From Lemma \ref{basic_lemma} we deduce :
\begin{equation*}
-\mathbf{n}^{\prime T} \frac{\partial\sigma}{\partial\mathbf{R}}\,\delta\mathbf{n}=\mathbf{n}^{\prime T}  \frac{\partial\sigma}{\partial\mathbf{R}}\mathbf{P}\left(\frac{\partial \boldsymbol{\zeta}}{\partial \mathbf{x}}\right)^T\mathbf{n},
\end{equation*} 
which proves Lemma \ref{key3}.
	\end{proof}
\section{Equations of motion and  shape equation}
The vesicle occupies domain ${D_t}$ with a free boundary  $S_{t}$ which is the membrane surface, and  $S_1$  which belongs to the rigid surface $\mathcal{S}=S_1\cup S_2$. $S_1$   denotes the footprint of $D_t$ on $\mathcal{S}$, and $C_t$ is the closed edge (contact line) between $S_1$ and $S_2$ (see Fig. 1). 
\\
We denote $\mathbf{n}_1$ the external unit normal to $S_1$ along  contact line $C_t$. Then denoting  $\mathbf{t}_1=-\mathbf{t}$, one has :
\begin{equation*}
\mathbf{n}'_1 = \mathbf{t}_1 \times \mathbf{n}_1=\mathbf{n}_1 \times \mathbf{t}. 
\end{equation*}
The surface energy of membrane  $S_{t}$  is denoted $\sigma$. Solid surfaces $S_{1}$ and $S_{2}$ have  constant
surface energies denoted $\sigma _{1}$ and $\sigma _{2}$. The geometrical notations
are shown  in  Fig. 1.
\begin{figure}[h]
\begin{center}
\includegraphics[width=9cm]{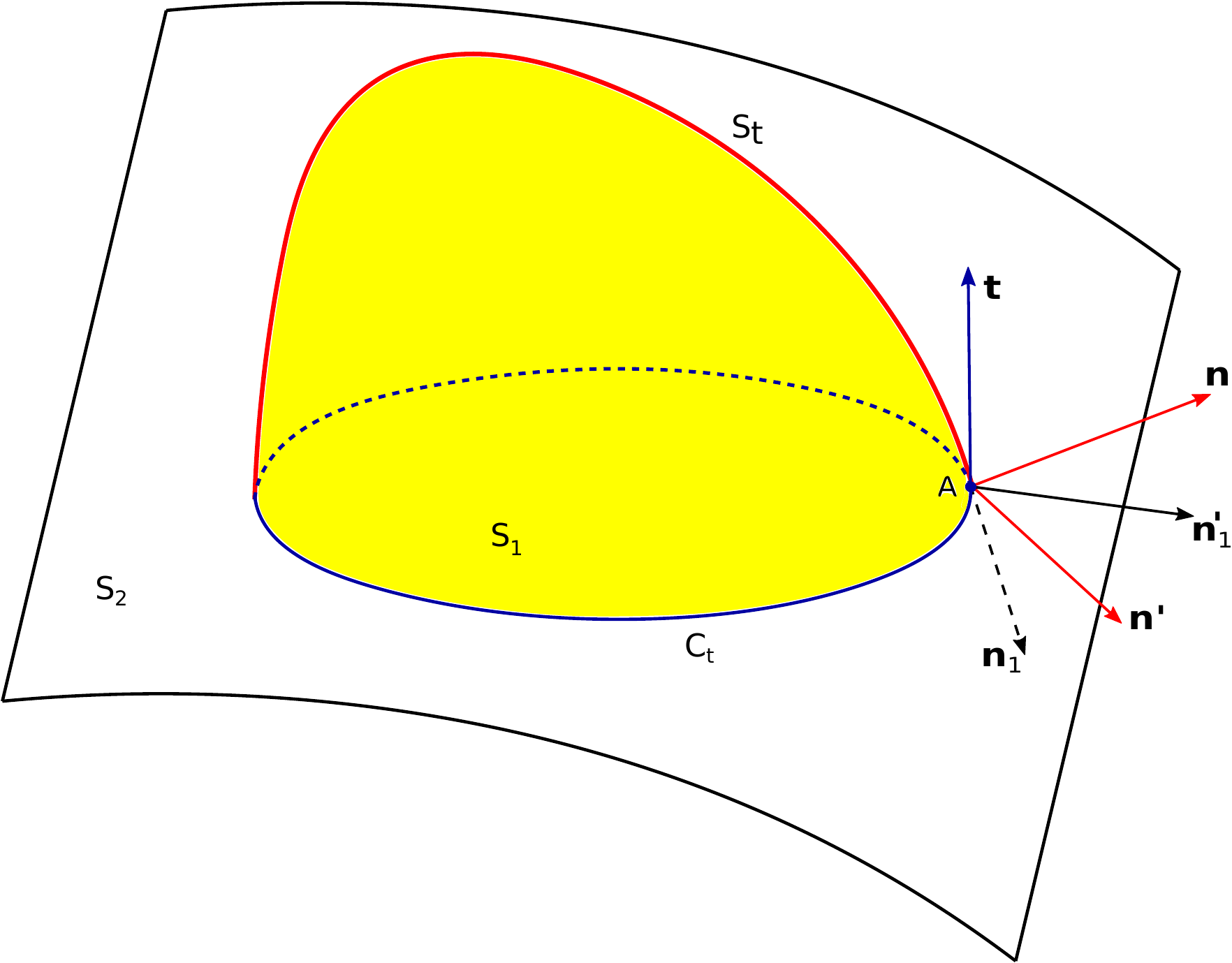}
\end{center}
\caption{ A drop  lies on   solid surface $\mathcal{S}=S_1\cup S_2$ ; $S_{t}$ is a free surface;  $\mathbf {n} _1$
and $\mathbf {n}$ are the external unit normal vectors to $S_1$  and $S_t$, respectively. Contact line $C_t$ separates $S_1$ and $S_2$,
$\mathbf{t} $ is the unit tangent vector to $C_t$ on $\mathcal{S}$. Vectors $\mathbf {n} ^{\prime}_1 = \mathbf {n} 
_{1}\times\mathbf {t} $ and $ \mathbf {n} ^{\prime}=
\mathbf {t} \times \boldsymbol{n}$ are the binormals to $C_t$
relatively to $\mathcal{S}$ and $S_{t}$ at point $A$ of $C_t$, respectively. }
\label{fig1}
\end{figure}

One can formulate the virtual work principle in the form  \cite{Germain_1973,Gouin_2007}:
\begin{equation*}
\delta{\cal A}_{e}+\delta{\cal A}_{i}-\delta {\cal E}=0,
\end{equation*}
where $\delta{\cal A}_{e}$ is the virtual work of external forces, $\delta{\cal A}_{i}$ is the virtual work of inertial forces, and $\delta {\cal E}$ is the variation of the total energy. 
The  energy  ${\cal E}$ is taken in  the form :
\begin{equation*}
 {\cal E}= \iiint_{D_t}\rho\,\varepsilon\, dv+ \iint_{S_t}\sigma\, ds+ \iint_{S_1}\sigma_1\, ds, 
\end{equation*}
where   specific internal energy $\varepsilon$ is a function of density $\rho$. As we mentioned before,  one can also include in this dependence several scalar quantities which are transported by the flow  (specific entropy, mass fractions of surfactants, etc.).
From Lemma  \ref{basic_lemma}, Eq. \eqref{Jacobi} and the mass conservation law :
\begin{equation*}
\rho\, {\rm det} \, \mathbf{F}= \rho_0(\mathbf{X}),
\end{equation*}
we obtain the variation of the specific energy and density at fixed Lagrangian coordinates in the form : 
\begin{equation*}
\delta\varepsilon = \frac{p}{\rho^2}\,\delta\rho\quad {\rm with}\quad \delta\rho= -\rho\,{\rm div}{\boldsymbol \zeta} ,
\end{equation*}
where $p$ is the thermodynamical pressure. 
Consequently, the variation of the first term is \cite{Berdichevsky,Gavrilyuk,Serrin}: 
\begin{eqnarray*}
	\delta \iiint_{D_t}\rho\,\varepsilon\, dv &=& \delta \iiint_{D_0}\rho_0\,\varepsilon\, dv_0= \iiint_{D_0}\rho_0\,\delta \varepsilon\, dv_0\\
	&=& \iiint_{D_t}\rho\,\delta\varepsilon\, dv=- \iiint_{D_t}p\ {\rm div}\,{\boldsymbol \zeta}\, dv.
	\end{eqnarray*}
The variation of the surface energy is  given in Lemma \ref{variation_surface_energy}. The third term is the surface energy of  $S_1$ with energy $\sigma_1$ per unit surface. 
The virtual work of the external forces is given in the form : 
\begin{equation*}
\delta{\cal A}_{e}=\iiint_{D_t}\rho\, \mathbf{f}^{T} {\boldsymbol \zeta}dv+\iint_{S_t}\mathbf{T}^T{\boldsymbol \zeta}ds+\int_{C_t}\sigma_2{\mathbf{n}_1 ^{\prime T}} \boldsymbol{\zeta}ds, 
\end{equation*}
where $\rho\,\mathbf f$ is the volume external force in $D_t$, $\mathbf{T}$ is the external stress vector at the free surface $S_t$, and $\sigma_2\, {\mathbf{n}_1^{\prime}}$ is the line tension vector exerted on $C_t$.  The last term on the right-hand side comes from  Lemma \ref{variation_surface_energy} which can be also applied for rigid surfaces. Finally, 
\begin{equation*}
\delta{\cal A}_{i}=-\iiint_{D_t} \rho\, \mathbf{a}^T {\boldsymbol \zeta}dv
\end{equation*}
is the virtual work of inertial force, where $\mathbf{a}$ is the acceleration.    The virtual work of forces $\delta \mathcal{T}$  applied to the material volume $D_t$ is defined as : 
\begin{eqnarray}
\delta \mathcal{T} &=&\displaystyle\iiint_{D_t}\left( -\rho\, \mathbf{a}^T +\rho\, \mathbf{f}^{T}-\mathbf{ 
}\boldsymbol{\nabla }^{T}p\right) \boldsymbol{\zeta \ }dv+\iint_{S_{1}}
\left( p+H_{1}\sigma _{1}\right) \mathbf{n}_{1}^{T}
 \boldsymbol{\zeta  }\,ds \notag\\
\displaystyle&+& \iint_{S_{t}}\bigg [\,-{\rm div_{tg}} \left(\frac{\partial\sigma}{\partial\mathbf{R}}\mathbf{R}\right)  +{\rm div_{tg}}\left(	\mathbf{P}\frac{\partial\sigma}{\partial\mathbf{R}}\right) \mathbf{R}\label{vwork}
\\
\displaystyle  & -&\, {\rm div_{tg}}\left(\mathbf{P}\,{\rm div_{tg}^{\it T}} \left(	\mathbf{P}\frac{\partial\sigma}{\partial\mathbf{R}}\right)\right)  \mathbf{n}^{T} +  \left( p+ H \sigma \right) \mathbf{n}%
^{T}+\boldsymbol{\nabla}_{\rm tg}^{T}\sigma   +\mathbf{T}^{T}\bigg]   \boldsymbol{\zeta}\,ds\notag
\\
\displaystyle  &- &\int_{C_t}\left\lbrace  \left[ \left( \sigma _{1}-\sigma _{2}\right) \mathbf{n}%
_{1}^{\prime T}+\sigma \,\mathbf{n}^{\prime T}-\mathbf{n}^{\prime T}{\rm div}_{\rm tg}^T\left(	\mathbf{P}\frac{\partial\sigma}{\partial\mathbf{R}}\right)\mathbf{n} ^T-\mathbf{n}^{\prime T}\frac{\partial\sigma}{\partial\mathbf{R}}{\mathbf{R}}\right] \boldsymbol{%
	\zeta }  \right.\notag\\
\displaystyle & &\qquad \ \ +\left.\mathbf{n}^{\prime T} \frac{\partial\sigma}{\partial\mathbf{R}} \,\mathbf{P}\left(\frac{\partial \boldsymbol{\zeta}}{\partial \mathbf{x}}\right)^T\mathbf{n}  \right\rbrace dl. \notag
\end{eqnarray}
As usually,   $H_{1}$ and $H$ are the sum of principle curvatures  of surfaces $S_{1}$ and  $S_{t}$, respectively.  
Terms on $D_t$, $S_{1}$,  $S_{t}$  are in separable form with respect to the
field $\boldsymbol{\zeta }$.  Expression  \eqref{vwork} implies the equation of  motion in $D_t$ and boundary conditions on  surfaces $S_{1}$,  $S_{t}$ 
  \cite{Schwartz}.  Virtual displacement  $\boldsymbol{\zeta }$ must be compatible with conditions of the problem; for example, $S_{1}$ is an external surface to domain $D_t$ and consequently $\boldsymbol{\zeta }$ must be tangent to $S_{1}$. This notion is developed in \cite{Berdichevsky}.
They are presented below. 

\subsection{Equation of motion}
We consider   virtual displacements $\boldsymbol{
\zeta }$ which vanish  on the boundary of   $D_t$. The fundamental lemma of virtual displacements yields :
\begin{equation}
\rho\,\mathbf{a}+{\nabla } p =\rho\, \mathbf{f}, \label{vwork2}\ 
\end{equation}
which  is the classical Newton law in continuum mechanics.

\subsection{Condition  on surface  $S_1$}
{\color{black} Due to the fact that the surface $S_1$ is - {\it a priori} - given, the virtual displacements must be compatible with the geometry of $S_1$. This means that   the non-penetration condition (slip condition) is  verified : 
	\begin{equation}
	\mathbf{n}_{1}^{T}
\boldsymbol{\zeta  }=0.
\label{slip}
\end{equation} 
Constraint \eqref{slip} is equivalent to the introduction of  a Lagrange multiplier $\mathcal{P}_1$ into \eqref{vwork} where $\boldsymbol{\zeta }$ is now a virtual displacement without  constraint.   The corresponding term on  $S_1$ will be  modified into 
\begin{equation*}
\iint_{S_{1}}
\left( p+H_{1}\sigma _{1}-\mathcal{P}_1\right) \mathbf{n}_{1}^{T}
\boldsymbol{\zeta  }\,ds.
\end{equation*}
Since the variation of  $\boldsymbol{\zeta  }$ on $S_1$ is independent, Eq. (\ref{vwork}) implies :
\begin{equation}
\mathcal{P}_1=p +H_{1}\sigma _{1}.\label{shape1}
\end{equation}
}
This is the classical Laplace condition allowing us to obtain the normal stress component   $\mathcal{P}_1 \mathbf{n}_{1}$  exerted by surface $S_1$.

\subsection{\it Extended shape equation}
Taking account of Eqs. \eqref{vwork2} and  \eqref{shape1}, for all displacement $\boldsymbol{\zeta  }$ on   moving membrane $S_t$, one has from  Eq. (\ref{vwork}) : 
\begin{eqnarray*}
	\displaystyle& & \iint_{S_{t}}\bigg [\,-{\rm div_{tg}} \left(\frac{\partial\sigma}{\partial\mathbf{R}}\,\mathbf{R}\right)  +{\rm div_{tg}}\left(	\mathbf{P}\,\frac{\partial\sigma}{\partial\mathbf{R}}\right) \mathbf{R} 
	\\
	\displaystyle  & -&\, {\rm div_{tg}}\left(\mathbf{P}\,{\rm div_{tg}^{\it T}} \left(	\mathbf{P}\frac{\partial\sigma}{\partial\mathbf{R}}\right)\right)  \mathbf{n}^{T} +  \left( p+ H \sigma \right) \mathbf{n}%
	^{T}+\boldsymbol{\nabla}_{\rm tg}^{T}\sigma   +\mathbf{T}^{T}\bigg]   \boldsymbol{\zeta}\,ds=0 .
\end{eqnarray*}
It implies :
\begin{eqnarray}
&&\left\{p+ H  \sigma -{\rm div_{tg}}  \left[\mathbf{P}\,{\rm div_{tg}^{\it T}} \left(	\mathbf{P}\frac{\partial\sigma}{\partial\mathbf{R}}\right)\right]\right\}   \mathbf{n}
\notag \\ && +{\nabla}_{\rm tg} \sigma-\,{\rm div_{tg}^{\it T}} \left(\frac{\partial\sigma}{\partial\mathbf{R}}\mathbf{R}\right)+ \mathbf{R}\,{\rm div_{tg}^{\it T}}\left(	\mathbf{P}\frac{\partial\sigma}{\partial\mathbf{R}}\right) +\mathbf{T}=0. \label{shape}
\end{eqnarray}
Equation \eqref{shape} is the  most general form of the dynamical boundary condition on $S_t$. Due to the fact that surface energy $\sigma$ must be an isotropic function of curvature tensor $\mathbf{R}$, {\it i.e.}   a function of two invariants $H$ and $K$,  we obtain (for proof, see Appendix) that the following vector
\begin{equation*}
 {\nabla}_{\rm tg} \sigma    - {\rm div_{tg}^{\it T}} \left(\frac{\partial\sigma}{\partial\mathbf{R}}\mathbf{R}\right)+ \mathbf{R}\,{\rm div_{tg}^{\it T}}\left(	\mathbf{P}\frac{\partial\sigma}{\partial\mathbf{R}}\right) \label{parterm}
\end{equation*}
is  normal to $S_t$ and consequently $\mathbf{T}$ writes in the form :
\begin{equation*}
\mathbf{T}= - \mathcal{P}\,\mathbf{n}.
\end{equation*}
Here scalar $\mathcal{P}$ has the  dimension of  pressure.\\
One obtains from Eq. \eqref{preshape} (see Appendix) :
 \begin{eqnarray}
\displaystyle H \sigma  -\Delta _{\mathrm{tg}}a &-& b\,\Delta _{\mathrm{tg}}H-\nabla _{\mathrm{tg%
}}^{T}b\,\nabla _{\mathrm{tg}}H-\mathrm{div_{tg}}\left( \mathbf{R}\,\nabla _{%
	\mathrm{tg}}b\right)  \notag \\ 
  \displaystyle &+& \left( 2K-H^{2}\right) \frac{\partial \sigma }{\partial H}-HK\frac{\partial
	\sigma }{\partial K}=\mathcal{P}-p.\label{Tangent_normal}
 \end{eqnarray}
Relation \eqref{Tangent_normal} is the normal component of Eq. \eqref{shape}.\\
It is important to underline that  equation (\ref{shape}) is only expressed in the normal direction to $S_t$. This is not  the case when surface energy $\sigma$ also depends on  physico-chemical characteristics  of $S_t$,  as  temperature or surfactants. In this last case, Marangoni effects can appear  producing  additive tangential terms to $S_t$.\\
 Using Lemma \ref{Lemma 1} (second equation)  and expressions of scalars $a$ and $b$ given by Eq. \eqref{lemmea}, we get the {\it extended shape equation}: 
\begin{eqnarray}
	&H&\left(\sigma-K\frac{\partial\sigma}{\partial K}\right) +\left(2K-H^2\right)\frac{\partial\sigma}{\partial H}-  \Delta_{\rm tg}\frac{\partial\sigma}{\partial H}-H\, \Delta_{\rm tg}\frac{\partial\sigma}{\partial K} \notag\\
	&-&\nabla _{\mathrm{tg}}^{T}H\, \nabla _{\mathrm{tg}}\frac{\partial\sigma}{\partial K}\,+\mathrm{div_{tg}}\left( \mathbf{R}\,\nabla _{%
		\mathrm{tg}}\frac{\partial\sigma}{\partial K}\right) =\mathcal{P}-p.\label{shape2}
\end{eqnarray} 
Equation \eqref{shape2} was also derived in \cite{virga} under the hypothesis (\ref{hypothesis}) and the assumption of inextensibility of the membrane. Our derivation does not use these hypotheses. For example, the inextensibility property is not natural  even in the case of incompressible fluids (at fixed volume, the surface of a 3D body may vary). 
\subsection{Helfrich's shape equation} 
\noindent The Helfrich energy is given by Eq. \eqref{Helfrich_energy}. The shape equation \eqref{shape2} immediately writes in the form :
\begin{equation}
\sigma_0\, H+ \frac{\kappa}{2} \left(H-H_0\right)\left[4K-H\left(H+H_0\right)\right]-\kappa\, \Delta_{\rm tg}\,H=\mathcal{P}-p\,,\label{Helfrich}
\end{equation}
which is the classical form obtained by Helfrich \footnote{Let us note that Helfrich  considered the vesicle as an incompressible fluid. He also assumed that the membrane  has a total constant area. Then, the virtual work can be  expressed as 
	\begin{equation*}
	\delta \mathcal{T} =\iiint_{D}\rho \ \boldsymbol{f}^{T}\boldsymbol{\zeta}%
	~dv+  \iint_{S} \ \boldsymbol{T}^{T}\boldsymbol{\zeta}~ds -\,\delta\iint_S \sigma~ds \, + \lambda_0 \, \delta\iint_S ds + \,  \delta   \iiint_D p\ \rm{div}\ \boldsymbol{\zeta}\ dv , 
	\end{equation*}
	where the scalar  $\lambda_0$ is a constant  Lagrange multiplier  and $p$\,   is a distributed  Lagrange multiplier. The 'shape equation' is similar to \eqref{Helfrich}.}.

\section{Extended Young-Dupr\'e condition on contact line $C_t$}
Let us denote by $\theta=\left\langle\mathbf{n}^\prime,\mathbf{n}^\prime_1 \right\rangle = \pi +\left\langle\mathbf{n},\mathbf{n}_1 \right\rangle [{\rm mod}\, 2\pi ] $ the Young angle between $S_1$ and $S_t$  (see Fig. \ref{fig2}).  
\begin{figure}[h]
	\begin{center}
		\includegraphics[width=12cm]{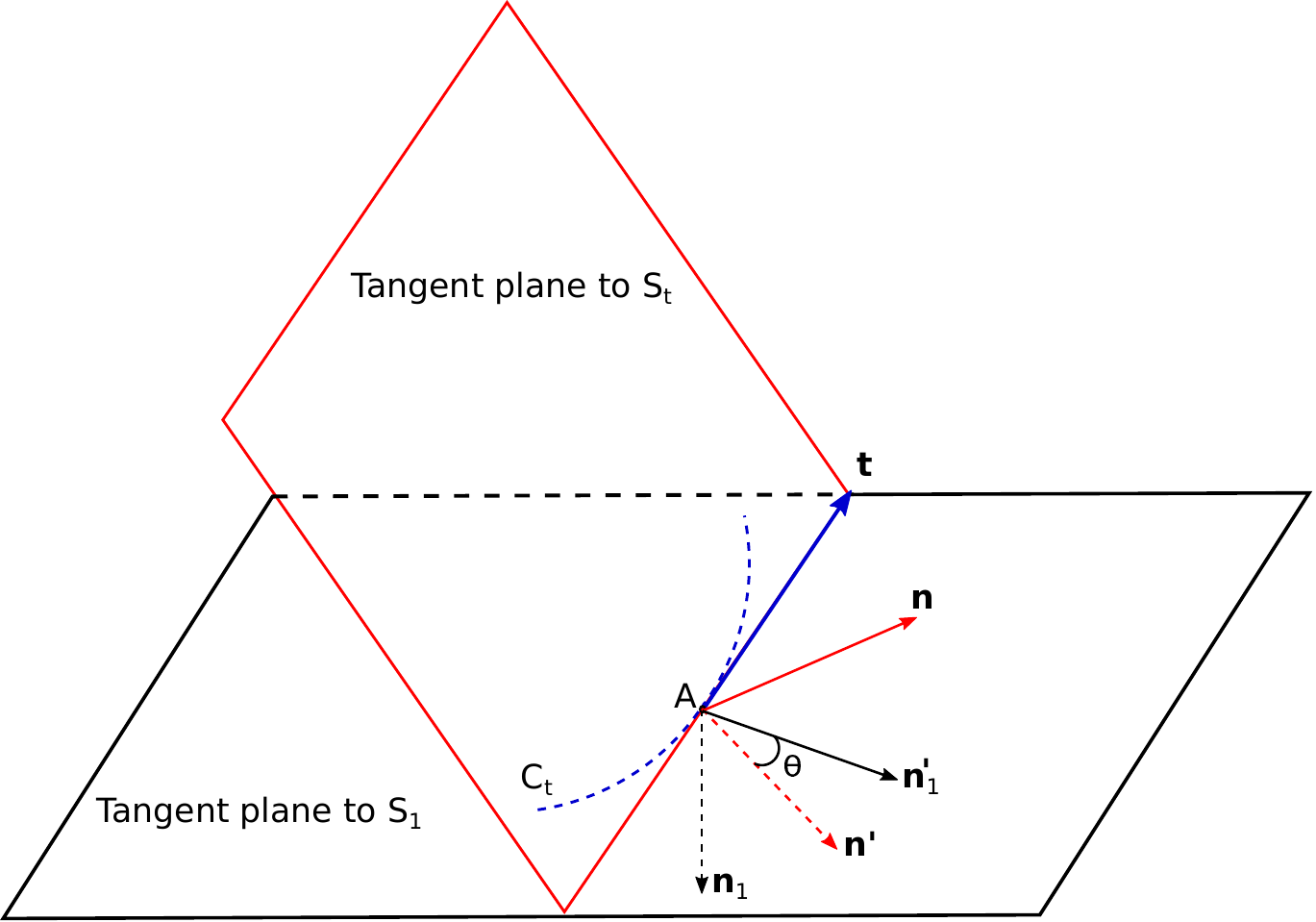}
	\end{center}
	\caption{Tangent planes to membrane $S_t$ and  solid surface $S_1$ : 
		\label{fig2}  $\mathbf {n} _1$
		and $\mathbf {n}$ are the unit normal vectors to $\mathcal{S}$  and $S_t$,
		external to the domain of the vesicle; contact line $C_t$ is shared between $\mathcal{S}$
		and $S_{t}$ and $\mathbf{t} $ is the unit tangent vector to $C_t$
		relatively to $\mathbf {n}$; $\mathbf {n} ^{\prime}_1 = \mathbf {n} 
		_{1}\times\mathbf {t} $ and $ \mathbf {n} ^{\prime}=
		\mathbf {t} \times \boldsymbol{n}$ are   binormals to $C_t$
		relatively to $\mathcal{S}$ and $S_{t}$  at point $A$, respectively. Angle $\theta = \left\langle \mathbf {n} ^{\prime},\mathbf {n} ^{\prime}_1\right\rangle $. {\color{black} The normal plane to $C_t$ at A  contains vectors  $\mathbf {n},  \mathbf {n} ^{\prime},   \mathbf {n}_1,  \mathbf {n}_1 ^{\prime}  $}.} 
\end{figure}

Due to the fact that  $C_t$ belongs to $S_1$, the virtual displacement on $C_t$  is in the form :
\begin{equation}
\boldsymbol{\zeta}=\alpha  \mathbf{t} + \beta \mathbf{n}^\prime_1,  \label{tangentdis}
\end{equation}
{\color{black} where $\alpha$ and $\beta$ are two scalar fields defined  on $S_1$.   Let us remark that condition \eqref{tangentdis} expresses the non-penetration condition \eqref{slip} on $S_1$. Moreover, since $\mathbf{n}$, $\mathbf{n}_1$, $\mathbf{n}_1^\prime$ belong to the normal plane to $C_t$ at $A$ (see Fig. \ref{fig2}), one has : 
\begin{equation}
\boldsymbol{n}= \boldsymbol{n}_1^{\prime}\sin\; \theta-\boldsymbol{n}_1\cos\; \theta.
\label{normal_relation}
\end{equation}
}
But relation $\boldsymbol{\zeta}^T\boldsymbol{n}_1 = 0 $  implies 
\begin{equation*}
\mathbf{P}\left(\frac{\partial \boldsymbol{\zeta}}{\partial \mathbf{x}}\right)^T\mathbf{n}_1+ \mathbf{P}\left(\frac{\partial \mathbf{n}_1}{\partial \mathbf{x}}\right)^T \boldsymbol{\zeta}=0.
\end{equation*}
Replacing \eqref{normal_relation} into \eqref{vwork} one has :
\begin{eqnarray}
& &	\displaystyle  \delta \mathcal{T}= - \int_{C_t}\left\lbrace  \left[ \left( \sigma _{1}-\sigma _{2}\right) \mathbf{n}%
	_{1}^{\prime T}+\sigma \,\mathbf{n}^{\prime T}-\mathbf{n}^{\prime T}{\rm div}_{\rm tg}^T\left(	\mathbf{P}\frac{\partial\sigma}{\partial\mathbf{R}}\right)\mathbf{n} ^T-\mathbf{n}^{\prime T}\frac{\partial\sigma}{\partial\mathbf{R}}{\mathbf{R}}\right] \boldsymbol{%
		\zeta }  \right.\notag\\
	\displaystyle & &\qquad \ \ +\left.\mathbf{n}^{\prime T} \frac{\partial\sigma}{\partial\mathbf{R}} \,\mathbf{P}\left(\frac{\partial \boldsymbol{\zeta}}{\partial \mathbf{x}}\right)^T\mathbf{n}  \right\rbrace dl = \notag\\
	\displaystyle  & & - \int_{C_t}\left\lbrace  \left[ \left( \sigma _{1}-\sigma _{2}\right) \mathbf{n}%
	_{1}^{\prime T}+\sigma \,\mathbf{n}^{\prime T}-\mathbf{n}^{\prime T}{\rm div}_{\rm tg}^T\left(	\mathbf{P}\frac{\partial\sigma}{\partial\mathbf{R}}\right)\mathbf{n} ^T-\mathbf{n}^{\prime T}\frac{\partial\sigma}{\partial\mathbf{R}}{\mathbf{R}}+{\color{black}\cos}\,\theta\,\mathbf{n}^{\prime T}\frac{\partial\sigma}{\partial\mathbf{R}}{\mathbf{P}}\,\left(\frac{\partial\mathbf{n}_1}{\partial\mathbf{x}} \right)^T
\right] \boldsymbol{%
		\zeta }  \right.\notag\\
	\displaystyle & &\qquad \ \ +\ {\sin}\,\theta\,\left.\mathbf{n}^{\prime T} \frac{\partial\sigma}{\partial\mathbf{R}} \,\mathbf{P}\left(\frac{\partial \boldsymbol{\zeta}}{\partial \mathbf{x}}\right)^T\mathbf{n}_1^{\prime}  \right\rbrace dl=0. \label{boundary}
\end{eqnarray}
{\color{black} We choose now the virtual displacement in the form $\boldsymbol{\zeta}=\beta \;  {\mathbf n}_1^{\prime}$. One has :
\begin{equation*}
\frac{\partial \boldsymbol{\zeta}}{\partial \mathbf{x}}={\mathbf n}_1^{\prime}
\; (\nabla \beta)^T+\beta \; \frac{\partial {\mathbf n}_1^{\prime}}{\partial \mathbf{x}}, \quad \left(\frac{\partial \boldsymbol{\zeta}}{\partial \mathbf{x}}\right)^T=\nabla \beta\; {\mathbf n}_1^{\prime T}
+\beta \; \left(\frac{\partial {\mathbf n}_1^{\prime}}{\partial \mathbf{x}}\right)^T.
\end{equation*}
Since $\displaystyle{\left(\frac{\partial{\mathbf n}_1^{\prime}}{\partial \mathbf{x}}\right)^T {\mathbf n}_1^{\prime}}=0$, it implies :
\begin{equation*}
\left(\frac{\partial \boldsymbol{\zeta}}{\partial \mathbf{x}}\right)^T{\mathbf n}_1^{\prime}=\nabla \beta. 
\end{equation*}
The integral \eqref{boundary} becomes :
\begin{eqnarray}
& &	\displaystyle 
\displaystyle \int_{C_t}\left\lbrace  \left[ \left( \sigma _{1}-\sigma _{2}\right) \mathbf{n}%
_{1}^{\prime T}+\sigma \,\mathbf{n}^{\prime T}-\mathbf{n}^{\prime T}{\rm div}_{\rm tg}^T\left(	\mathbf{P}\frac{\partial\sigma}{\partial\mathbf{R}}\right)\mathbf{n} ^T-\mathbf{n}^{\prime T}\frac{\partial\sigma}{\partial\mathbf{R}}{\mathbf{R}}+{\color{black}\cos}\,\theta\,\mathbf{n}^{\prime T}\frac{\partial\sigma}{\partial\mathbf{R}}{\mathbf{P}}\,\left(\frac{\partial\mathbf{n}_1}{\partial\mathbf{x}} \right)^T
\right] \mathbf {n }_1^{\prime}\; \beta  \right.\notag\\
\displaystyle & &\qquad \ \ +\ {\sin}\,\theta\,\left.\mathbf{n}^{\prime T} \frac{\partial\sigma}{\partial\mathbf{R}} \,\mathbf{P}\; \nabla\beta  \right\rbrace dl=0. 
\label{boundary1}
\end{eqnarray}
Since  $\beta$ and the components of $\nabla \beta$  can be choosen   as independent, 
relation \eqref{boundary1}  implies two boundary conditions. The first condition on line $C_t$  is : 
\begin{equation}
{\rm sin}\,\theta\, \mathbf{n}^{\prime T} \frac{\partial\sigma}{\partial\mathbf{R}} \,\mathbf{P} =0. 
\label{bending}
\end{equation}
The second condition is : 
\begin{equation}
\left[ \left( \sigma _{1}-\sigma _{2}\right) \mathbf{n}%
_{1}^{\prime T}+\sigma \,\mathbf{n}^{\prime T}-\mathbf{n}^{\prime T}{\rm div}_{\rm tg}^T\left(	\mathbf{P}\frac{\partial\sigma}{\partial\mathbf{R}}\right)\mathbf{n} ^T-\mathbf{n}^{\prime T}\frac{\partial\sigma}{\partial\mathbf{R}}{\mathbf{R}}+{\color{black}\cos}\,\theta\,\mathbf{n}^{\prime T}\frac{\partial\sigma}{\partial\mathbf{R}}{\mathbf{P}}\,\left(\frac{\partial\mathbf{n}_1}{\partial\mathbf{x}} \right)^T
\right] \mathbf {n }_1^{\prime}=0.
\label{preliminary_condition}
\end{equation}

The case ${\rm sin}\,\theta=0$ all along $C_t$ is degenerate. If $\theta=0$, this corresponds to a hydrophobic surface (the contact line is absent). If $\theta =\pi$, this corresponds to a complete wetting. In the last case ${\mathbf n}_1^\prime=-{\mathbf n}^\prime$, ${\mathbf n}_1={\mathbf n}$, and the condition \eqref{preliminary_condition} becomes trivial : $\sigma_1-\sigma_2-\sigma=0$.

The general case corresponds to the partial wetting ($\sin\; \theta\ne 0$). Due to Eq. \eqref{R},
\begin{equation*}
\mathbf{n}^{\prime T} \frac{\partial\sigma}{\partial\mathbf{R}} \,\mathbf{P} \equiv\mathbf{n}^{\prime T} ( a \;\mathbf{I}+b\,
\mathbf{R}) \mathbf{P}\equiv  a\, \mathbf{n}^{\prime T} +b\, \mathbf{n}^{\prime T} \mathbf{R}\equiv\mathbf{n}^{\prime T} \frac{\partial\sigma}{\partial\mathbf{R}}. 
\end{equation*}
Hence,  Eq. \eqref{bending}   yields
\begin{equation}
\mathbf{n}^{\prime T} \frac{\partial\sigma}{\partial\mathbf{R}}=0 \label{bending1}.
\end{equation}
}
Equation \eqref{bending1} implies (see Lemma \ref{derivative_with_respect_to_R}) : 
\begin{equation*}
\mathbf{n}^{\prime T} \left[\left(\frac{\partial \sigma}{\partial H}+ H \frac{\partial \sigma}{\partial K}\right)\mathbf{I}- \frac{\partial \sigma}{\partial K}\mathbf{R}\right]=0.
\end{equation*}
Consequently,   $\mathbf{n}^{\prime}$ is an eigenvector of $\mathbf{R}$. We denote  $c_{n^\prime}$ the associated eigenvalue $c_2$. Then
\begin{equation}
\frac{\partial \sigma}{\partial H}+ H \frac{\partial \sigma}{\partial K}= c_{n^\prime}\frac{\partial \sigma}{\partial K}.
\label{clamping_1}
\end{equation}
Due to the fact that $\mathbf{t}$ is also  eigenvector of $\mathbf{R}$ with eigenvalue $c_t=c_1$ ($\mathbf{t}$ and $\mathbf{n}^{\prime}$ form the eigenbasis of $\mathbf{R}$ along $C_t$),  we get $H=c_t+c_{n^\prime}$ and the equivalent to the boundary condition \eqref{clamping_1}  in the form :
\begin{equation}
\label{clamping}
\frac{\partial \sigma}{\partial H}+ c_t \frac{\partial \sigma}{\partial K}= 0.
\end{equation} 
From Lemma \ref{derivative_with_respect_to_R}, Eq. \eqref{lemmea}, we immediately deduce :
\begin{equation}
{\rm div_{tg}}\left(	\mathbf{P}\frac{\partial\sigma}{\partial\mathbf{R}}\right) =\nabla_{\rm tg}
^{T}a+\left(a\,H+b\,H^2-2\,b\,K\right)   \mathbf{n}^{T} +\nabla_{\rm tg}
^{T}b\  \mathbf{R}+b\, \nabla_{\rm tg}
^{T}H . 
\label{key4}
\end{equation}
Due to the fact that $\mathbf{n}^{\prime T} \mathbf{n} =0$, we obtain :
\begin{equation*}
\mathbf{n}^{\prime T} \,{\rm div}_{\rm tg}^T\left(	\mathbf{P}\frac{\partial\sigma}{\partial\mathbf{R}}\right) =\mathbf{n}^{\prime T}  \left[\nabla_{\rm tg}
a  +   \mathbf{R}\,\nabla_{\rm tg}
b+b\, \nabla_{\rm tg}H\right]=\mathbf{n}^{\prime T}  \left[\nabla 
a  +   \mathbf{R}\,\nabla 
b+b\, \nabla H\right] 
\end{equation*}
 Consequently, one obtains the second condition on $C_t$ in the form  : 
 \begin{equation}
\sigma _{1} -\sigma _{2}  
+\sigma \, {\rm cos}\,\theta {\color{black}-} {\rm sin}\,\theta\, \mathbf{n}^{\prime T} \left(\nabla a+ b\, \nabla H+  \mathbf{R}\,\nabla b\right)= 0.
 \label{Young1}
 \end{equation}
This is {\it the extended Young-Dupr\'e condition} along contact line $C_t$ between membrane $S_t$ and solid surface $\mathcal{S}$  
 (\footnote {\color{black} The virtual displacement   taken in  the most general form   \eqref{tangentdis}  does not produce new boundary conditions.  Due to the linearity of the virtual work, to prove this property  it is sufficient to take  $ \boldsymbol{\zeta}= \alpha\, \mathbf{t}$. We obtain
\begin{equation*}
\frac{\partial \boldsymbol{\zeta}}{\partial \mathbf{x}}={\mathbf t} 
 (\nabla \alpha)^T+\alpha\; \frac{\partial {\mathbf t}}{\partial \mathbf{x}}, \qquad \left(\frac{\partial \boldsymbol{\zeta}}{\partial \mathbf{x}}\right)^T{\mathbf n}_1^\prime= 
 \alpha\; \left(\frac{\partial {\mathbf t}}{\partial \mathbf{x}}\right)^T {\mathbf n}_1^\prime. 
 \end{equation*}
 Since
 \begin{equation*}
 \frac{\partial {\mathbf t}}{\partial \mathbf{x}} = c\,  {\mathbf N}\, {\mathbf t}^T,
\end{equation*}
where ${\mathbf N}$ is the principal unit normal and $c$ is the curvature along $C_t$, one obtains :
\begin{equation*}
  \left(\frac{\partial \boldsymbol{\zeta}}{\partial \mathbf{x}}\right)^T{\mathbf n}_1^\prime= 
\alpha\;c\; {\mathbf t}\, {\mathbf N}^T \, {\mathbf n}_1^\prime  
\end{equation*}
and
\begin{equation*}
\sin \,\theta\, \mathbf{n}^{\prime T} \frac{\partial\sigma}{\partial\mathbf{R}} \,\mathbf{P}\left(\frac{\partial \boldsymbol{\zeta}}{\partial \mathbf{x}}\right)^T\mathbf{n}_1^{\prime} =\alpha\; c\;\sin \,\theta\,\mathbf{n}^{\prime T} \frac{\partial\sigma}{\partial\mathbf{R}}\;{\mathbf t}\, {\mathbf N}^T \, {\mathbf n}_1^\prime,
\end{equation*}
which is equal to zero thanks to Eq. \eqref{bending1}.\\
Moreover, thanks to Eq. \eqref{bending1}, we immediately obtain that  term
\begin{equation*}
 \left[ \left( \sigma _{1}-\sigma _{2}\right) \mathbf{n}%
_{1}^{\prime T}+\sigma \,\mathbf{n}^{\prime T}-\mathbf{n}^{\prime T}{\rm div}_{\rm tg}^T\left(	\mathbf{P}\frac{\partial\sigma}{\partial\mathbf{R}}\right)\mathbf{n} ^T-\mathbf{n}^{\prime T}\frac{\partial\sigma}{\partial\mathbf{R}}{\mathbf{R}}+ \cos\,\theta\,\mathbf{n}^{\prime T}\frac{\partial\sigma}{\partial\mathbf{R}}{\mathbf{P}}\,\left(\frac{\partial\mathbf{n}_1}{\partial\mathbf{x}} \right)^T
\right] \mathbf {t } \; \alpha
\end{equation*}
is vanishing.  Hence, new boundary conditions do not appear on  $C_t$.
 }).
\\ 
  In the case of  Helfrich's energy given by relation \eqref{Helfrich_energy}, we obtain the extended Young-Dupr\'e condition \eqref{Young1} in the form :
 \begin{equation}
	\sigma _{1}  -  \sigma _{2}  
	+ \sigma  \, {\rm cos}\,\theta {\color{black}-}  \kappa\,{\rm sin}\,\theta\,\mathbf{n}^{\prime T}\nabla H   = 0 .  \label{YoungHelfrich}
\end{equation}
This last condition was previously obtained in \cite{Gouin_2014_vesicle}.
\section{Surfaces of revolution}
\subsection{Shape equation for the surfaces of revolution}

\noindent Along a revolution surface, the invariants of the curvature tensor depend only on $s$ which is the curvilinear abscissa 
of meridian curve denoted by $\Gamma$ \cite{Aleksandrov} : 
 \begin{equation*}
 H=H(s), \quad K=K(s).
 \end{equation*} 
 \begin{figure}[h]
 	\begin{center}
 		\includegraphics[width=12cm]{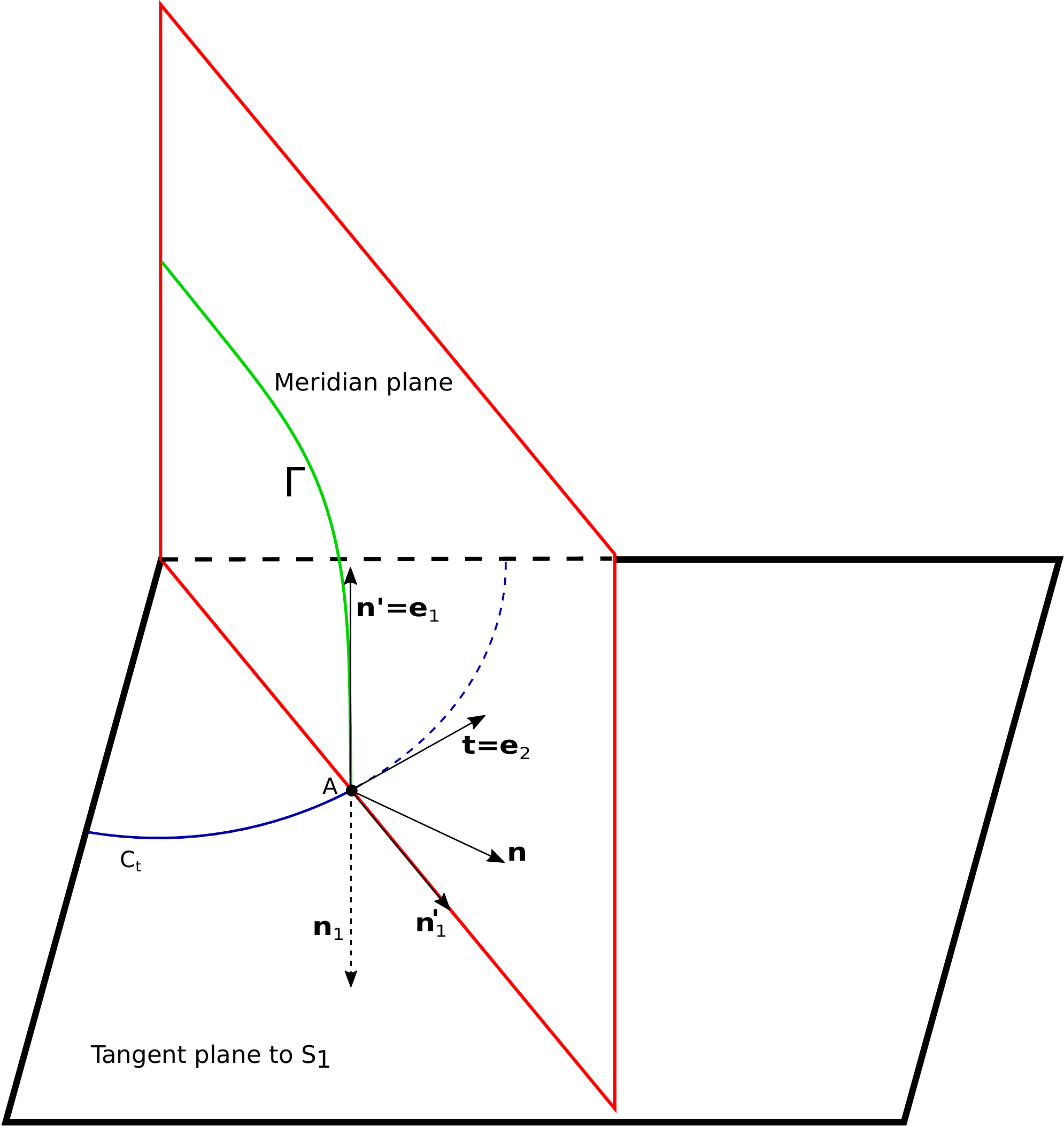}
 	\end{center}
 	\caption{The case of a revolution domain. The line $C_t$  (contact edge between $S_t$ and $S_1$) is a circle with an axis which is the revolution axis collinear to ${\mathbf{n}}_1$. The meridian curve is denoted $\Gamma$; normal vector ${\mathbf{n}}$ and binormal vector   ${\mathbf{n}}^\prime$  are in the meridian plane of   revolution surface  $S_t$. We have $\mathbf{n}^\prime= \mathbf{e}_1$ and $\mathbf{t}= \mathbf{e}_2$, corresponding to the eigenvectors of the curvature tensor $\mathbf R$ at $A$.}
 	\label{fig3}
 \end{figure} 
 One of the eigenvectors, denoted ${\mathbf{e}}_1$, of the curvature tensor  $\mathbf R$ is tangent to  meridian curve  $\Gamma$ (see Fig. \ref{fig3}).
Let us remark that for any function $f(s)$, one has :
\begin{equation*}
\nabla_{\mathrm{tg}}\,f=\frac{df}{ds}\,\mathbf{e}_1, \quad \Delta_{\mathrm{tg}} \, f=\frac{d^2f}{ds^2}.
\end{equation*}
Indeed, the first equation is the  definition of the tangential gradient. The second equality is obtained as follows :
\begin{equation*}
{\rm div}_{\mathrm{tg}} \left(\frac{df}{ds} {\mathbf e}_1\right) ={\rm tr}\left({\mathbf P}\frac{\partial }{\partial {\mathbf x}}\left(\frac{df}{ds}{\mathbf e}_1\right)\right)
={\rm tr}\left({\mathbf P}\frac{d}{ds}\left(\frac{df}{ds}{\mathbf e}_1\right)\otimes{\mathbf e}_1\right)
\end{equation*} 
\begin{equation*}
={\rm tr}\left(\frac{d^2f}{ds^2}{\mathbf P}{\mathbf e}_1\otimes{\mathbf e}_1+c_1(s)\frac{df}{ds}{\mathbf n}\otimes{\mathbf e}_1\right)=\frac{d^2f}{ds^2}.
\end{equation*} 
The Fr\'enet formula was used here :
\begin{equation*}
\frac{d{\mathbf e}_1}{ds}=c_1{\mathbf n}.
\end{equation*} 
Also, 
\begin{equation*}
{\rm div}_{\mathrm{tg}} \left(\mathbf{R}\nabla_{\mathrm{tg}}\,f\right) ={\rm div}_{\mathrm{tg}} \left(\frac{df}{ds}\mathbf{R}{\mathbf e}_1\right)={\rm div}_{\mathrm{tg}} \left(\frac{df}{ds}c_1{\mathbf e}_1\right)
=\frac{d}{ds}\left(c_1\displaystyle{\frac{df}{ds}}\right).
 \end{equation*} 
  For surfaces of revolution  the shape equation \eqref{shape2} becomes :
 \begin{eqnarray*}
 &H&\left(\sigma-K\frac{\partial\sigma}{\partial K}\right) +\left(2K-H^2\right)\frac{\partial\sigma}{\partial H}-\frac{d^2}{ds^2}\left(\frac{\partial\sigma}{\partial H}\right)-H\, \frac{d^2}{ds^2}\left(\frac{\partial\sigma}{\partial K}\right) \notag\\
 &-&\frac{dH}{ds}\frac{d}{ds}\left(\frac{\partial\sigma}{\partial K}\right)+\frac{d}{ds}\left(c_1\frac{d}{ds}\left(\frac{\partial\sigma}{\partial K}\right)\right) =\mathcal{P}-p.
 \end{eqnarray*}

\subsection{Extended  Young--Dupr\'e condition for surfaces of revolution} 
One has along $C_t$, ${\mathbf t} ={\mathbf e}_2$, ${\mathbf n}^\prime ={\mathbf e}_1$. It implies ${\mathbf n}^{\prime T}\mathbf{R}\, {\mathbf t} =0$. Also, one has :
\begin{equation*} 
\mathbf{n}^{\prime T} \left(\nabla a+ b\, \nabla H+  \mathbf{R}\,\nabla b\right)=\frac{da}{ds}+b\,\frac{dH}{ds}+c_1\,\frac{db}{ds}.
\end{equation*}
The Young -- Dupr{\'e} condition \eqref{Young1} becomes :
 \begin{equation*} 
\sigma_1-\sigma_2\,{\rm cos}\, \theta{\color{black}-}{\rm sin}\,\theta \left(\frac{da}{ds}+b\,\frac{dH}{ds}+c_t\,\frac{db}{ds} \right)  = 0.
\end{equation*}
Since 
 \begin{equation*} 
a=\frac{\partial \sigma}{\partial H}+H\frac{\partial \sigma}{\partial K}, \quad b=-\frac{\partial \sigma}{\partial K},
\end{equation*}
one finally obtains :
 \begin{equation*} 
\sigma_1-\sigma_2\,{\rm cos}\, \theta {\color{black}-}{\rm sin}\,\theta \left[\frac{d}{ds}\left(\frac{\partial \sigma}{\partial H}\right)+c_{n^\prime} \frac{d}{ds}\left(\frac{\partial \sigma}{\partial K}\right)\right]  = 0.
\end{equation*}
For the Helfrich energy \eqref{Helfrich_energy} this expression yields :
\begin{equation*} 
\sigma_1-\sigma_2\,{\rm cos}\, \theta {\color{black}-}\kappa \frac{dH}{ds}  \, {\rm sin}\,\theta    = 0.
\end{equation*}

\section{Conclusion}
 Membranes can be considered as material surfaces   endowed with a surface energy density depending on the invariants of the curvature tensor : $\sigma=\sigma(H,K)$.  By using the principle of virtual working, we derived the boundary conditions on the moving membranes (``shape equation") as well as two boundary conditions  on the contact line.   In  limit cases, we recover classical boundary conditions. The ``shape equation" and the boundary conditions are summarized below in the non-degenerate case (see (\ref{shape2}), (\ref{clamping}), (\ref{Young1})) as  
 \begin{itemize}
 	 {\color{black}
\item {\it the equation for the moving surface $S_t$  }:\\

$
	 \bullet\quad  \displaystyle  H \left(\sigma-K\frac{\partial\sigma}{\partial K}\right) +\left(2K-H^2\right)\frac{\partial\sigma}{\partial H}-  \Delta_{\rm tg}\frac{\partial\sigma}{\partial H}-H\, \Delta_{\rm tg}\frac{\partial\sigma}{\partial K} $\\
	 
	$ \quad\ \ \displaystyle  -\, \nabla _{\mathrm{tg}}^{T}H\, \nabla _{\mathrm{tg}}\frac{\partial\sigma}{\partial K}\,+\mathrm{div_{tg}}\left( \mathbf{R}\,\nabla _{%
		\mathrm{tg}}\frac{\partial\sigma}{\partial K}\right) =\mathcal{P}-p. 
$ \\

\item {\it the clamping condition  on the moving line $C_t$}  :\\

$ \bullet \quad \displaystyle\frac{\partial \sigma}{\partial H}+ c_t \frac{\partial \sigma}{\partial K}= 0,
$
\\

Also,  ($\mathbf{t}$, $\mathbf{n}$, $\mathbf{n}^\prime$) - which is the {\it Darboux frame} - are the  eigenvectors of  curvature tensor ${\mathbf R}$.
\vskip 0.3cm

\item {\it dynamic generalization of the Young-Dupr\'{e} condition on $C_t$} :\\

$
\bullet \quad \sigma _{1} -\sigma _{2}  
 +\sigma \, {\rm cos}\,\theta {\color{black}-} {\rm sin}\,\theta\, \mathbf{n}^{\prime T} \left(\nabla_{\rm tg}\left(\displaystyle\frac{\partial \sigma}{\partial H}\right)+\left(H{\mathbf P}-{\mathbf R}\right)\nabla_{\rm tg}\displaystyle\left(\frac{\partial \sigma}{\partial K}\right)\right)=0.
$} 
\end{itemize}
 {In the case of Helfrich's energy the generalization of Young-Dupr\'{e} condition is reduced to equation \eqref{YoungHelfrich}: 
 \begin{equation*}
 \sigma _{1} -\sigma _{2}  
 +\sigma \, {\rm cos}\,\theta {\color{black}-} \kappa \;  {\rm sin}\,\theta \;  \mathbf{n}^{\prime T}\nabla_{\rm tg} H=0.
 \end{equation*}
The last term,  corresponding to the variation of the mean curvature of $S_t$ in the binormal direction at the contact line, can dominate the other terms.  It could be interpreted as a line tension term usually added in the models with  constant surface energy (cf. \cite{Babak}). It should also be noted that the droplet volume  has no effect in  the classical Young-Dupr\'e condition. This is not the case for the generalized Young-Dupr\'e condition since the curvatures can become very large for very small droplets  (they are inversely proportional to the droplet  size). 
The clamping condition for the Helfrich energy fixes the value of $H$ on the contact line :
\begin{equation*}
H=H_0-c_t\frac{\bar \kappa}{\kappa}.
\end{equation*}
The new shape equation and boundary conditions can  be used for solving  dynamic problems. This could be, for example,  the study of the ``fingering" phenomenon appearing as a result of the non-linear instability of a moving contact line. This complicated problem will be studied  in the future.  
}\\

\textbf{Acknowledgments}. The authors thank the anonymous referees for their  noteworthy and helpful remarks that greatly contributed to improve the final version of the paper.

 \section{Appendix}
Since $\sigma= \sigma(H,K)$, we get :
 \begin{equation}
 \nabla_{\rm tg} \sigma= \frac{\partial\sigma}{\partial H} \nabla_{\rm tg} H + \frac{\partial\sigma}{\partial 
 	K} \nabla_{\rm tg} K .\label{Gradsigma}
 \end{equation}
 From Eq. \eqref{lemmea}, we obtain :
 \begin{equation}
 {\rm div_{tg}}\left(	\mathbf{P}\frac{\partial\sigma}{\partial\mathbf{R}}\right) =\nabla_{\rm tg}
 ^{T}a+\left(a\,H+b\,H^2-2\,b\,K\right)   \mathbf{n}^{T} +\nabla_{\rm tg}
 ^{T}b\  \mathbf{R}+b\, \nabla_{\rm tg}
 ^{T}H . 
 \label{key4}
 \end{equation}
 Also, one has :
 \begin{equation*}  
 {\rm div_{tg}}\left(\frac{\partial\sigma}{\partial\mathbf{R} 	}\,\mathbf{R}\right) ={\rm div_{tg}}\left( a\, 	\mathbf{R}\right)+{\rm div_{tg}}\left( b\, 	\mathbf{R}^2\right).
 \end{equation*}
 Due to \eqref{5}, one has :
 \begin{eqnarray*}
 	&&{\rm div_{tg}}\left(a\, 	\mathbf{R}\right)= (\nabla_{\rm tg}^T a)\	\mathbf{R}
 	+a\, \nabla_{\rm tg}^TH+ a\left(H^2-2\,K\right)	\mathbf{n}^T,\\
 	&&{\rm div_{tg}}\left(b\, 	\mathbf{R}^2\right) =    {\rm div_{tg}} [b\left(H	\mathbf{R}-K\mathbf{P}\right)]\\
 	&&  =    \nabla_{\rm tg}^T( b H)\	\mathbf{R} +b\,H\left[\nabla_{\rm tg}^T H+\left(H^2-2\,K\right)	\mathbf{n}^T \right]   -  \nabla_{\rm tg}^T (b K) -b\,K  H  \mathbf{n}^T       .               \end{eqnarray*}
 Consequently,
 \begin{eqnarray}
 {\rm div_{tg}}\left(\frac{\partial\sigma}{\partial\mathbf{R} 	}\,\mathbf{R}\right) &= &
 \left(\nabla_{\rm tg}^T (a+b\,H)\right)\,\mathbf{R} \label{key5}\\
 &+&  (a+b\,H)\,\nabla_{\rm tg}^T H - \nabla_{\rm tg}^T( bK) + (aH^2+bH^3-2aK-3bHK)\mathbf{n}^T . \notag
 \end{eqnarray}
 From relations \eqref{Gradsigma}, \eqref{key4}, \eqref{key5}, we deduce :
 \begin{equation*}
 {\nabla}_{\rm tg} \sigma   - {\rm div_{tg}^{\it T}} \left(\frac{\partial\sigma}{\partial\mathbf{R}}\mathbf{R}\right)+ \mathbf{R}\,{\rm div_{tg}^{\it T}}\left(	\mathbf{P}\frac{\partial\sigma}{\partial\mathbf{R}}\right) =   \left( 2\,a K+3\,b H K-a\,H^2-b\,H^3\right)  \mathbf{n}.\label{normalterm}
 \end{equation*}
Using \eqref{key4}, one obtains :
\begin{equation*}
\mathbf{P}\,{\rm div_{tg}^{\it T}} \left(	\mathbf{P}\frac{\partial\sigma}{\partial\mathbf{R}}\right) =\nabla_{\rm tg}\, a+\mathbf{R}\,\nabla_{\rm tg}\, b+b\, \nabla_{\rm tg}\, H. 
\end{equation*}
One deduces : 
\begin{equation}
{\rm div_{tg}}  \left[\mathbf{P}\,{\rm div_{tg}^{\it T}} \left(	\mathbf{P}\frac{\partial\sigma}{\partial\mathbf{R}}\right)\right] =\Delta_{\rm tg}  a\, + {\rm div_{tg}}\left({\mathbf{R}\,\nabla_{\rm tg}  b }\right)+ \, b\, \Delta_{\rm tg} H+{\nabla_{\rm tg}^T} b\, \nabla_{\rm tg} H. \label{key6}
\end{equation}
 From relations \eqref{Gradsigma}, \eqref{key4},  \eqref{key5}, we deduce : 
 \begin{eqnarray*}
 	&&{\nabla}_{\rm tg} \sigma   - {\rm div_{tg}^{\it T}} \left(\frac{\partial\sigma}{\partial\mathbf{R}}\mathbf{R}\right)+ \mathbf{R}\,{\rm div_{tg}^{\it T}}\left(	\mathbf{P}\frac{\partial\sigma}{\partial\mathbf{R}}\right) = \\
 	&&\left( 2\,a\,K+3\,b\,H\,K-a\,H^2-b\,H^3\right)  \mathbf{n}\notag\\
 	&& + \frac{\partial\sigma}{\partial H}\, \nabla_{\rm tg}H +\frac{\partial\sigma}{\partial K}\, \nabla_{\rm tg}K-\mathbf{R}\,\nabla_{tg} \left(a+b\,H\right)- \left(a+b\,H\right)\, \nabla_{\rm tg}H +\, \nabla_{\rm tg}(b\,K)\notag\\
 	&&+\mathbf{R}\,\nabla_{tg} a+\left(a\,H +b\,H^2- 2\,b\,K\right) \mathbf{R}\,\mathbf{n}+\mathbf{R}^2\,\nabla_{tg}b+b\,\mathbf{R} \,\nabla_{tg}H +\mathbf{T}={\mathbf 0}.  
 \end{eqnarray*}
 Using relations $\mathbf{R}\,\mathbf{n}=\mathbf{0}$, Eq. \eqref{Lemma 1}$_3$ and expressions of $a$ and $b$ given by Eq. \eqref{lemmea}, we obtain :
 \begin{eqnarray*}
 	&&  \frac{\partial\sigma}{\partial H}\, \nabla_{\rm tg}H +\frac{\partial\sigma}{\partial K}\, \nabla_{\rm tg}K-\mathbf{R}\,\nabla_{tg} \left(a+b\,H\right)- \left(a+b\,H\right)\, \nabla_{\rm tg}H +\, \nabla_{\rm tg}(b\,K)\notag\\
 	&&+\mathbf{R}\,\nabla_{tg} a+\left(a\,H +b\,H^2- 2\,b\,K\right) \mathbf{R}\,\mathbf{n}+\mathbf{R}^2\,\nabla_{tg}b+b\,\mathbf{R} \,\nabla_{tg}H =  \mathbf{0}.
 \end{eqnarray*}
 Consequently,  
 \begin{equation*}
 {\nabla}_{\rm tg} \sigma   - {\rm div_{tg}^{\it T}} \left(\frac{\partial\sigma}{\partial\mathbf{R}}\mathbf{R}\right)+ \mathbf{R}\,{\rm div_{tg}^{\it T}}\left(	\mathbf{P}\frac{\partial\sigma}{\partial\mathbf{R}}\right) =   \left( 2\,a K+3\,b H K-a\,H^2-b\,H^3\right)  \mathbf{n}.\label{normalterm}
 \end{equation*}
 Finally, using \eqref{key6}, one obtains : 
 \begin{eqnarray}
 &&        \quad \left[p+H \sigma-\Delta_{\rm tg}a - b\,  \Delta_{\rm tg}H-\nabla_{\rm tg}^T b\ \nabla_{\rm tg}H - {\rm div_{tg}}\left(\mathbf{R}\,\nabla_{tg} b\right)\right. \notag\\ &&\left.+\left( 2\,a\,K+3\,b\,H\,K-a\,H^2-b\,H^3\right)  \right]\mathbf{n}+\mathbf{T}={\mathbf{0}}, \label{preshape} 
 \end{eqnarray}
 where all tangential terms disappear in the boundary condition on $S_t$.

\end{document}